\documentclass[a4, 12pt]{article}

\usepackage{amsmath, amssymb, amsfonts}
\usepackage{float, graphicx, epsfig, epstopdf}
\usepackage{commath, mathtools}
\usepackage{algpseudocode}

\usepackage{natbib}
\bibpunct{(}{)}{;}{a}{,}{,}

\usepackage{amsthm}
\newtheorem{theorem}{Theorem}
\newtheorem{remark}{Remark}
\newtheorem{notation}{Notation}
\newtheorem{definition}{Definition}
\newtheorem{lemma}{Lemma}
\newtheorem{corollary}{Corollary}

\usepackage[affil-it]{authblk} 
\usepackage{etoolbox}
\usepackage{lmodern}

\newcommand{\slfrac}[2]{\left.#1\middle/#2\right.}
\newcommand{\fV}{\tilde{V}}
\newcommand{\fI}{\tilde{I}}
\newcommand{\fL}{\tilde{L}}
\newcommand{\defeq}{\vcentcolon=}

\begin{document}
\begin{flushleft}
{\Large
\textbf{A sparse reformulation of the Green's function formalism allows efficient simulations of morphological neuron models}
}
\\
Willem A.M. Wybo$^{1}$, 
Daniele Boccalini$^{2}$, 
Benjamin Torben-Nielsen$^{1,3}$,
Marc-Oliver Gewaltig$^{1}$
\\
\textbf{1} Blue Brain Project, Brain Mind Institute, EPFL, Geneva, Switzerland
\\
\textbf{2} Chair of Geometry, Mathematics Institute for Geometry and Applications, EPFL, Lausanne, Switzerland
\\
\textbf{3} Computational Neuroscience Unit, Okinawa Institute of Science and Technology, Okinawa, Japan
\\
$\ast$ E-mail: willem.wybo@epfl.ch
\end{flushleft}


\begin{abstract}
We prove that when a class of partial differential equations, generalized from the cable equation, is defined on tree graphs, and when the inputs are restricted to a spatially discrete, well chosen set of points, the Green's function (GF) formalism can be rewritten to scale as $O(n)$ with the number $n$ of input locations, contrary to the previously reported $O(n^2)$ scaling. We show that the linear scaling can be combined with an expansion of the remaining kernels as sums of exponentials, to allow efficient simulations of equations from the aforementioned class. We furthermore validate this simulation paradigm on models of nerve cells and explore its relation with more traditional finite difference approaches. Situations in which a gain in computational performance is expected, are discussed.
\end{abstract}


\section{Introduction}
Neurons have extensive morphological ramifications, called dendrites, that receive and integrate inputs from other neurons, and then transmit the result of this integration to the soma, or cell body, where an output in the form of action potentials is generated. Dendritic integration is considered a hallmark of neuronal computation \citep{London2005, Hausser2003, Silver2010} and parallels the dendritic morphology \citep{Torben-Nielsen2010, Agmon-Snir1998, Segev2000}. It is often studied using the cable equation --- a one dimensional reaction-diffusion equation that governs the evolution of the membrane potential $V(x,t)$, and is defined on a tree graph representing the dendritic arborization. Inputs, such as synaptic currents, are usually concentrated in a spatially discrete set of points and depend on this potential through their driving force \citep{Kuhn2004}, and possibly through other non-linear conductances such as NMDA \citep{Jahr1990}.

Traditionally, this integration is modeled using compartmental simulations, where space is discretized in a number of compartments of a certain length and a second order finite difference approximation is used to model the longitudinal currents \citep{Hines1984}. As the error of this approximation depends on the distance step, complex neural structures require many compartments and are computationally costly to simulate. 

Often however, stretches of neural fiber behave approximately linear \citep{Schoen2012}, and one is not interested in the explicit voltage at all locations, but only at a specific output location. For this reason we proposed the idea of using the Green's function formalism to simulate neuron models that receive inputs at a limited number of locations \citep{Wybo2013}. Given a number of $n$ inputs, the voltage at an output location $x_i$ can be written in this formalism as:
\begin{equation}\label{eq:GFsimple}
V(x_i, t) = \sum_{j=1}^{n} \int_{0}^{t} g(x_i, x_j, t-s) I_j(s) \mathrm{d}s.
\end{equation}
Problems arise when these inputs depend on the local voltage. Since, in such a case, all local voltages have to be known, a system of $n$ Volterra integral equations has to be integrated:
\begin{equation}\label{eq:GFhard}
V(x_i, t) = \sum_{j=1}^{n} \int_{0}^{t} g(x_i, x_j, t-s) I_j(s, V(x_j,s)) \mathrm{d}s, \hspace{4mm} i=1,\hdots,n.
\end{equation}
It can be seen that this system contains $n^2$ convolutions. This unfavorable scaling, along with the fact that the convolutions themselves are costly to compute and the restriction to point-source non-linearities, significantly impedes the computational efficiency of the GF formalism, and restricts its usefulness to very small numbers of input locations \citep[page 59-60]{Koch1998}.

In this work, we are able to significantly improve computational efficiency compared to the classic GF formalism by showing that all three perceived disadvantages can be overcome. Using a transitivity property for the Green's function \citep{Koch1998} (see appendix \S\ref{app:trans}), we show that, when the input locations are well chosen, a transformation of the system \eqref{eq:GFhard} exists so that only $O(n)$ kernels are required. We term this the \emph{sparse Green's function} (SGF) formalism. As an example of how an efficient integration algorithm can be designed for the resulting system of Volterra integral equations, we show that the kernels can be expressed as sums of exponentials using the vector fitting (VF) algorithm \citep{Gustavsen1999} (see appendix \S\ref{app:VF}) and that consequently the convolutions can be computed recursively. Finally, we show (in a simplified setting) that when the spacing between the input locations becomes sufficiently small, the SGF formalism reduces to the second order finite difference approximation (in the spatial component) of the original equation. As a consequence, the SGF formalism can be seen as a `generalization' of the second order finite difference approximation to arbitrary distance steps, as long as what lies in between the distance steps is approximately linear.

We validate this novel SGF formalism by reproducing two canonical results in neuroscientific modeling. First, we reproduce the result of \citep{Moore1978} on axonal action potential velocity. Second, we compare our SGF formalism with the de facto standard \textsc{neuron} simulator \citep{Carnevale2006} in the case of dendritic integration with conductance based synapses. In a final section, we discuss in which cases the SGF formalism may yield computational advantage over canonical second order finite difference approaches (of which the \textsc{neuron} simulator is an example).

\subsection*{The system of equations}
Each edge of the tree graph represents a segment of the dendritic tree, for which the cable equation has the following form:
\begin{equation}\label{eq:cable}
2 \pi a c_m \frac{\partial V}{\partial t}(x,t) +\frac{\pi a^2}{r_a} \frac{\partial^2 V}{\partial x^2}(x,t) - 2 \pi a g_m V(x,t) + \sum_c I_c(x,t) = \sum_{i=1}^n I_i(t,V(x_i,t)) \delta(x-x_i),
\end{equation}
where $c_m, g_m$ and $r_a$ denote, respectively, the membrane capacitance, the membrane conductance and the axial resistance, $a$ denotes the radius of the dendritic branch, $I_c$ the current contribution of a channel type $c$ and $I_i$ the input current at location $x_i$. The ion channel  current can depend non-linearly on the voltage and a number of state-variables:
\begin{equation}\label{eq:Ichan}
I_c(x,t) = f_c(V(x,t), \mathbf{y}_c(x,t)),
\end{equation}
where the state-variables $\mathbf{y}_c(x,t)$ evolve according to:
\begin{equation}\label{eq:chan}
\dot{y}_{c,j}(x,t) = \frac{y_{c, j,\text{inf}}(V(x,t)) - y_{c,j}(x,t)}{\tau_{c,j}(V(x,t))},
\end{equation}
with $\tau_{c,j}$ and $y_{c, j,\text{inf}}$ functions that depend on the channel type. Linearizing these currents yields a quasi-active description \citep{Koch1998} of the ion channels:
\begin{equation}\label{eq:Ichanqa}
I_{c,\text{lin}}(x,t) = \frac{\partial f_c}{\partial V}V(x,t) + \sum_j \frac{\partial f_c}{\partial y_{c,j}} y_{c,j}(x,t),
\end{equation}
with
\begin{equation}
\dot{y}_{c,j}(x,t) = \frac{d}{dV}\left(\frac{y_{c, j,\text{inf}}}{\tau_{c,j}}\right) V(x,t) - \frac{1}{\tau_{c,j}} y_{c,j}(x,t),
\end{equation}
where all derivatives, as well as $\tau_{c,j}$, are evaluated at the equilibrium values of the state variables. If there are a total of $K$ state-variables associated with ion channels, a system of $K+1$ PDE's of first degree in the temporal coordinate is obtained, which can be recast into a single PDE of degree $K+1$.

Consequently, we are interested in the GF of PDE's of the following form:
\begin{equation}\label{eq:PDE}
 \hat{L}(x) V(x,t) = \left[ \hat{L}_0(x) \frac{\partial^2}{\partial x^2} + \hat{L}_1(x) \frac{\partial}{\partial x} + \hat{L}_2(x) \right] V(x,t) = \hat{L}_3(x) \sum_{i=1}^n I_i(t) \delta(x-x_i),
\end{equation}
where $\hat{L}_i(x), i=0,\hdots,3$ are  operators of the form $\hat{L}_i(x) = \sum_{k = 0}^{K+1} C_k \frac{\partial^k}{\partial t^k}$, $K \in \mathbb{N}$, and $\delta$ is the Dirac-delta function.  
We assume:
\begin{enumerate}
\item[\textit{(i)}] that an equation of the form \eqref{eq:PDE} is defined on each edge of the tree graph (let $E$ denote the set of edges). 
\item[\textit{(ii)}] that on each leaf (let $\Lambda$ denote the set of leafs) a boundary condition of the following form holds:
\begin{equation} \label{eq:leafcond}
\hat{L}_{1}^{\lambda} \frac{\partial}{\partial x} V^{\epsilon_{\lambda}}(t) + \hat{L}_{2}^{\lambda} V^{\epsilon_{\lambda}}(t) = \hat{L}_{3}^{\lambda} I^{\lambda}(t) \hspace{4mm} \forall \lambda \in \Lambda,
\end{equation}
where $\hat{L}_{i}^{\lambda}$ are operators defined analogously to $\hat{L}_i(x)$ and where $V^{\epsilon_{\lambda}}$ is the field value on the adjacent edge $\epsilon_{\lambda}$ in the limit of $x^{\epsilon_{\lambda}}$ approaching the leaf. Note that in this general form, equation \eqref{eq:leafcond} can represent sealed end ($\hat{L}_{1}^{\lambda}=1, \hat{L}_{2}^{\lambda}=\hat{L}_{3}^{\lambda}=0$) or voltage clamp ($\hat{L}_{1}^{\lambda}=0, \hat{L}_{2}^{\lambda}=\hat{L}_{3}^{\lambda}=1$ and $I(t)$ constant) boundary conditions or, when there is only one neurite leaving the soma, a lumped soma boundary condition (see \citep{Tuckwell1988Introduction}, and where the operators can possibly contain higher order derivatives if quasi-active channels are present). 
\item[\textit{(iii)}] that at each node that is not a leaf (let $\Phi$ denote the set of nodes that are not leafs, and let $E(\phi)$ denote the set of edges that join at node $\phi \in \Phi$):
\begin{equation} \label{eq:nodecond}
\begin{aligned}
& V^{\epsilon}(t) = V^{\epsilon'}(t), \hspace{4mm} \forall \epsilon, \epsilon' \in E(\phi), \, \forall \phi \in \Phi \\
& \sum_{\epsilon \in E(\phi) }\hat{L}^{\epsilon} \frac{\partial}{\partial x} V^{\epsilon}(t) = \hat{L}_1^{\phi} V^{\epsilon}(t) 
 + \hat{L}_2^{\phi} I^{\phi}(t), \hspace{4mm} \forall \phi \in \Phi,
\end{aligned}
\end{equation}
where the operators $\hat{L}^{\epsilon, \phi}$ are again defined as above and $V^{\epsilon}$ denotes the voltage on edge $\epsilon$ in the limit of $x^{\epsilon}$ approaching the node $\phi$.  The first condition then expresses the continuity of the voltage at a node, whereas the second condition can signify the conservation of current flow ($\hat{L}^{\epsilon}=\slfrac{\pi a^2_{\epsilon}}{r_a^{\epsilon}}$,  $\hat{L}_1^{\phi}=\hat{L}_2^{\phi}=0$) or a somatic boundary condition when multiple neurites join at the soma (see again \citep{Tuckwell1988Introduction}).
\end{enumerate}
Algorithms to compute the GF of this system of PDE's have been described extensively in the neuroscientific literature: the algorithm by \citep{Koch1985} computes the Green's function exactly in the Fourier domain whereas the `sum-over-trips' approach pioneered by \citep{Abbott1991} (see \citep{Bressloff1997} for another overview) and extended by \citep{Coombes2007, Caudron2012} uses a path integral formalism. We implemented the algorithm given in \citep{Koch1985} as we are interested in the GF in the Fourier domain.


\section{Methods}
\subsection*{A sparse reformulation of the Green's function formalism}
In this section we prove formally that a sparse reformulation of equation \eqref{eq:GFhard} exists.
Fourier transforming this equation gives the GF formalism in the Fourier domain \citep{Butz1974, Wybo2013}, :
\begin{eqnarray}\label{eq:fourierGF}
\left( \begin{array}{c}
\fV_1(\omega) \\ \vdots \\ \fV_n(\omega)
\end{array} \right) =
\left( \begin{array}{ccc}
g_{11}(\omega) & \cdots & g_{1n}(\omega) \\
\vdots & \ddots & \vdots \\
g_{n1}(\omega) & \cdots & g_{nn}(\omega) \\
\end{array} \right)
\left( \begin{array}{c}
\fI_1(\omega) \\ \vdots \\ \fI_n(\omega)
\end{array} \right),
\end{eqnarray}
where $\fV_1, \hdots, \fV_n$ denote the field values resulting from the inputs $\fI_1,\hdots, \fI_n$ at locations $1,\hdots,n$ (which may be arbitrarily distributed along the edges of the tree graph), and $g_{ij}$ denotes the GF evaluated between points $i$ and $j$ (see \citep{Koch1985} for an example of an algorithm to evaluate these functions). Note that for the GF, we dropped the $\sim$ that signifies the Fourier transform for notational simplicity. Whenever the argument of a kernel is not explicitly mentioned, it will be implied that it is the Fourier transform that is under consideration. Note furthermore that the GF between points $i$ and $j$ is equal to the GF between points $j$ and $i$, so that the matrix in equation \eqref{eq:fourierGF} is symmetric \citep[page 63]{Koch1998}. Note finally that while the input current $\fI_i$ may depend on the local voltage (cf. equation \eqref{eq:cable}), it is still possible to take the Fourier transform by considering $I_i(t, V(x_i,t))$ as an a priori unknown function of time ($\equiv I_i(t)$).

We can rewrite the set of equations \eqref{eq:fourierGF} so that the field at one location depends on the input only at that location and the field at all other locations:
\begin{equation}\label{eq:rewrite}
\fV_i(\omega) = f_i(\omega)\fI_i(\omega) + \sum_{j\neq i} h_{ij}(\omega) \fV_j(\omega),
\end{equation}
where (with $G$ denoting the matrix for which $G_{ij} = g_{ij}$ and $G^{-1}$ its matrix inverse and omitting the argument $\omega$ for notational clarity)
\begin{eqnarray}\label{eq:inverse}
f_i & = & 1 / (G^{-1})_{ii} \nonumber \\
h_{ij} & = & - (G^{-1})_{ij} / (G^{-1})_{ii}.
\end{eqnarray}

Intuitively, the field at any location can only depend on the local input and the fields at the neighboring locations. After introducing the necessary definitions and notations, we will prove this intuition formally.

\begin{notation}\label{not:colrowdel}
Let $A$ denote a matrix. $A[j;i]$ then denotes the same matrix with the $j$'th row and $i$'th column deleted.
\end{notation}

This way, $(-1)^{j+i}\det(A[j;i])$ gives the $(j,i)$'th minor of $A$ and the elements of the adjugate matrix of $A$ can be written as:
\begin{equation}
\text{adj}(A)_{ij} =  (-1)^{j+i}\det(A[j;i])
\end{equation}
Then, by using Cramer's rule (see for instance \citep[pages 17-24]{Horn2012}):
\begin{equation}
A^{-1} = \frac{\text{adj}(A)}{\det(A)},
\end{equation}
equation \eqref{eq:inverse} can be written by as:
\begin{eqnarray}\label{eq:inverse2}
f_i & = & \det(G) / \det(G[i;i]) \nonumber \\
h_{ij} & = & (-1)^{i+j+1} \det(G[j;i]) / \det(G[i;i]).
\end{eqnarray}

We specify the discrete set of input locations as follows:

\begin{notation}
$\mathcal{P}$ denotes a set of $n$ points distributed on a tree graph.
\end{notation}

The geometry of the tree graph will induce a nearest neighbor relation on $\mathcal{P}$. We define this relation in the following way:

\begin{definition}\label{def:nn}
Two points $i,j \in \mathcal{P}$ are \emph{nearest neighbors} if no other point lies on the shortest path between them.
\end{definition}

This allows for the definition of sets of nearest neighbors:

\begin{definition}\label{def:setnn}
A \emph{set of nearest neighbors} $\mathcal{N}$ is a subset of $\mathcal{P}$ in which each pair of points is a pair of nearest neighbors, and for which no other point can be found in $\mathcal{P}$ that is not in $\mathcal{N}$ but is still a nearest neighbor of every point in $\mathcal{N}$.
\end{definition}

In Fig~\ref{fig:formalism}A we show an illustration of these sets.

\begin{figure}
   \centering
   \includegraphics[width=0.65\textwidth]{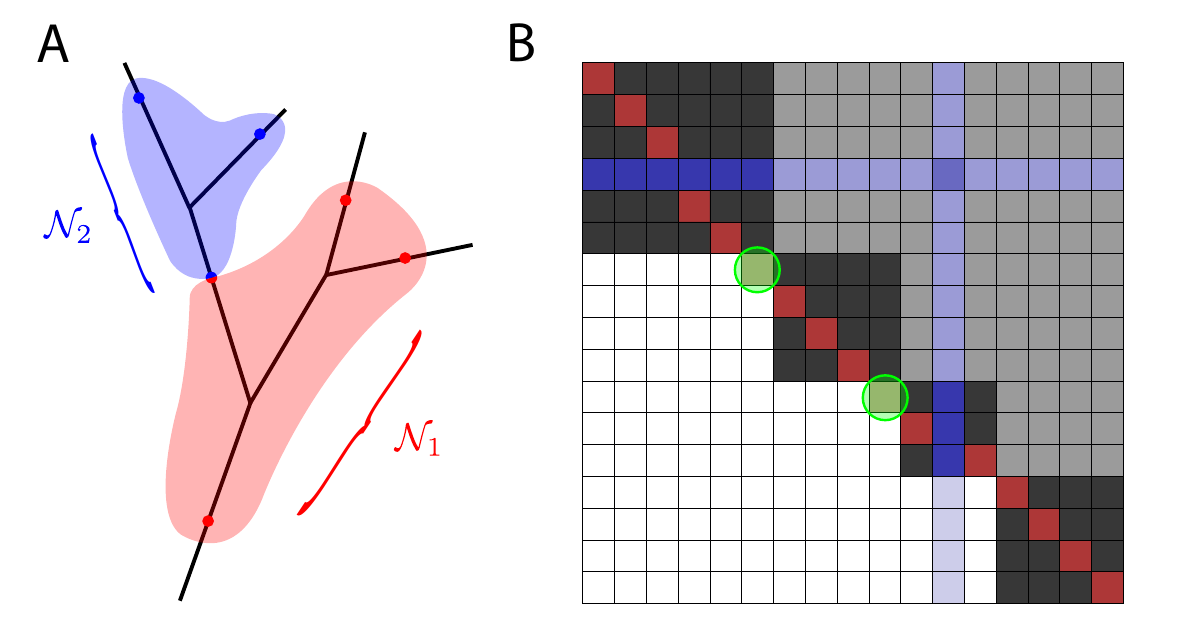}
   \caption{Schematic of the rationale behind the sparsification of the Green's function formalism for tree structures. A: Illustration of the sets of nearest neighbors $\mathcal{N}_q$ ($q=1,2$) induced by the tree structure. For any pair of points in $\mathcal{N}_{q}$ there is no other point on the shortest path between them. B: Illustration of how the structure of the matrices $A$ and $A[j,i]$ (here $j<i$) gives rise to the sparseness of the SGF. These matrices can be written in `block upper triangular form' (white: zeros, black: blocks on the diagonal, grey: other non-zero elements). The diagonal of $A[j,i]$ (red) is shifted between the $j$'th row and $i$'th column (blue) compared to the diagonal of $A$. When $j$ and $i$ do not belong to the same set of nearest neighbors, there is always at least one zero on the diagonal of the resulting matrix (here, there are two zeros --- green circles).}
   \label{fig:formalism}
\end{figure}

\begin{notation}\label{not:numsetnn}
We use $m$ to denote the number of sets of nearest neighbors which can be found in a given $\mathcal{P}$.
\end{notation}

\begin{definition}\label{def:matrestr}
For any set of points $\mathcal{P} = \{1,\hdots,n\}$, $G(\mathcal{P})$ is the matrix of transfer kernels:
\begin{eqnarray}
G(\mathcal{P}) =
\left( \begin{array}{cccc}
g_{11} & g_{12} & \cdots & g_{1n} \\
\vdots & \vdots & \ddots & \vdots \\
g_{n1} & g_{n2} & \cdots & g_{nn} \\
\end{array} \right),
\hspace{4mm} 1,\hdots,n \in \mathcal{P}
\end{eqnarray}
\end{definition}
However, when $\mathcal{P}$ denotes a set of input locations distributed on the tree graph, we will only make the argument of $G$ explicit when we consider a subset of $\mathcal{P}$. Hence it is understood that $G \equiv G(\mathcal{P})$, as is the case in formulas \eqref{eq:inverse} and \eqref{eq:inverse2}.

\begin{definition} \label{def:att}
An \emph{attenuation function} $a_{ij}$ is defined by:
\begin{equation}
a_{ij} = g_{ij} / g_{jj}.
\end{equation}
\end{definition}
Note that trivially $a_{ii} = 1$.

\begin{definition}\label{def:attmatrestr}
$A(\mathcal{P})$ is the matrix of attenuation functions $a_{ij}$ between any two points $i,j \in \mathcal{P}$ (analogous to definition \ref{def:matrestr}).
\end{definition}

The significance of these sets $\mathcal{N}$ (definition~\ref{def:setnn}) is due to the the following lemma \citep[page 63]{Koch1998}:

\begin{lemma}\label{thm:transitivity}
\textbf{Transitivity property.} If $i,j \in \mathcal{P}$ are not nearest neighbors, then for every point $l$ on the direct path between them it holds that:

\begin{equation}\label{eq:trans}
g_{ij} =  \frac{g_{il}g_{lj}}{g_{ll}},
\end{equation}
\end{lemma}

\begin{proof}
See appendix \S\ref{app:trans}.
\end{proof}

This leads directly to a similar transitivity property for attenuations:

\begin{equation}\label{eq:transat}
a_{ij} = a_{il}a_{lj}.
\end{equation}

Note that we may obtain the matrix $A$ (definition~\ref{def:attmatrestr}) from $G$ by dividing, for all $i=1,\hdots,n$, its $i$'th column by $g_{ii}$, and that as a consequence $\left(\prod_{i} g_{ii} \right) \det(A) = \det(G)$. Hence \eqref{eq:inverse2} can be written as:
\begin{eqnarray}\label{eq:inverse3}
f_i & = & g_{ii} \det(A) / \det(A[i;i]) \nonumber \\
h_{ij} & = & (-1)^{i+j+1} \det(A[j;i]) / \det(A[i;i]).
\end{eqnarray}

Given these definitions and notation, we can capture the aforementioned intuition formally as follows:

\begin{theorem}\label{thm:maintheorem}
Consider an arbitrary set of points $\mathcal{P}$ on a tree graph. Then
\begin{itemize}
\item[i)] for a point $i$ that is in exactly $p$ sets $\mathcal{N}_{1}, \hdots, \mathcal{N}_{p}$,
\begin{equation}\label{eq:Kin}
f_i = g_{ii}\frac{\det(A(\mathcal{N}_{1})) \cdots \det(A(\mathcal{N}_{p}))}{\det(A(\mathcal{N}_{1} \cup \hdots \cup \mathcal{N}_{p})[i;i])}
\end{equation}
\item[ii)] 
for a pair of points $i,j$ that are nearest neighbors, and for which there consequently exists a set $\mathcal{N} \ni i,j$, and where there are exactly $p$ other sets  $\mathcal{N}_{1}, \hdots, \mathcal{N}_{p}$ that contain $i$:
\begin{equation}\label{eq:Ktrans}
h_{ij} = (-1)^{i+j+1}\frac{\det(A(\mathcal{N}_{1})) \hdots \det(A(\mathcal{N}_{p})) \det(A(\mathcal{N})[j;i])}{\det(A(\mathcal{N}_{1} \cup \hdots \cup \mathcal{N}_{p} \cup \mathcal{N})[i;i])}
\end{equation}
\item[iii)] for a pair of points $i,j$ that are not nearest neighbors,
\begin{equation}\label{eq:stuff}
h_{ij} = 0 \hspace{3mm} \text{ and } \hspace{3mm} h_{ji} = 0.
\end{equation}
\end{itemize}
\end{theorem}

\begin{remark} \normalfont
To unclutter the notation, we use the indices $i,j \in \mathcal{N}$ to refer to the actual points in $\mathcal{P}$, their corresponding positions in the full matrices $A$ and in the restricted matrices $A(\mathcal{N})$. This is justified, as it does not influence formulas \eqref{eq:Kin} and \eqref{eq:Ktrans}. Permuting the numbering of a pair of points amounts to permuting the corresponding rows and columns of the matrices, and thus does not change the determinants. Restricting $A$ to $A(\mathcal{N})$ amounts to deleting the appropriate rows and corresponding columns from $A$, and thus does not change the factor $(-1)^{i+j+1}$.
\end{remark}

\begin{proof}
For convenience, we introduce the following ordering scheme for the points: we choose one point as the root of the tree and give it the lowest index, and then reorder the other points such that within each set $\mathcal{N}$, the point closest to the root has the lowest index, and the other points in $\mathcal{N}$ are numbered consecutively. We then reorder the sets $\mathcal{N}$ so that the set with highest index contains the point with highest index, and so on.

Let there be a total of $m$ sets of nearest neighbors $\mathcal{N}_q$, $q=1,\hdots,m$ in $\mathcal{P}$. First, we show that
\begin{equation}\label{eq:determ2}
\det(A) = \prod_{q=1}^m \det(A(\mathcal{N}_q))
\end{equation}
by applying elementary row operations recursively set after set. Let's start with $\mathcal{N}_m$, the last set of nearest neighbors, containing $k$ elements. Our numbering scheme guarantees that the last $k-1$ rows of $A$ correspond to the $k-1$ elements in $\mathcal{N}_m$ that are not closest to the root. Let us denote the point in $\mathcal{N}_m$ that is closest to the root by $l$. Between every $d \notin \mathcal{N}_m$ and every $e \in \mathcal{N}_m$ the transitivity property \eqref{eq:transat} holds, and thus $a_{e d} = a_{e l}a_{l d}$.

By applying row operations $R_e(A) \rightarrow R_e(A) - a_{el} R_l(A)$, that do not change $\det(A)$ \citep[pages 9-10]{Horn2012}, and using \eqref{eq:transat} along with the trivial identity $a_{el} = a_{el}a_{ll}$, the last $k - 1$ rows of $A$ become:
\begin{align}\label{eq:row}
& R_j(A) = \nonumber \\
& \left(
\begin{array}{cccccc}
0 & \hdots & 0 & a_{e,n-k+2} - a_{el}a_{l,n-k+2} & \hdots & a_{en} - a_{el}a_{ln}
\end{array} \right).
\end{align}
Thus, the determinant reduces to
\begin{equation}\label{eq:determinants}
\det(A) = \det \left(
\begin{array}{cc}
A' & A'' \\
0 & B
\end{array} \right)
 = \det(A') \det(B),
\end{equation}
where $A'$ and $A''$ are the parts of $A$ unchanged by the row operations, and $B$ consists of the non-zero elements of the part of $A$ affected by the row operations. It holds that
\begin{equation}
\det(B) = \det(A(\mathcal{N}_m)).
\end{equation}
Construct $A(\mathcal{N}_m)$ by taking the last $k-1$ rows and columns from $A$, and the attenuations to (from) point $l$ as the first row (column) of $A(\mathcal{N}_m)$. Then, by applying the row-operations $R_e(A(\mathcal{N}_m)) \rightarrow R_e(A(\mathcal{N}_m)) - a_{e1} R_1(A(\mathcal{N}_m))$, the matrix $B$ will be found as the only non-zero minor of the first row elements --- the determinant of which is multiplied by 1 to give the determinant of $A(\mathcal{N}_m)$. Thus we have factorized $\det(A)$ in $\det(A')\det(A(\mathcal{N}_m))$.

Our numbering scheme for the points guarantees that a similar reduction can be applied to $\det(A')$, as long as it contains distinct sets $\mathcal{N}$, which by induction on $m$ proves \eqref{eq:determ2}. Note that we could have achieved the same reduction by applying the column operations $C_e(A) \rightarrow C_e(A) - a_{le} C_l(A)$ in an analogous manner. Then the matrix in \eqref{eq:determinants} would have its zero part in the upper right corner.

Next we compute the determinant of $A[i;i]$. By removing both the $i$'th row and column, it is as if $i$ is removed from the set $\mathcal{P}$ completely. Consequently, $\mathcal{N}_{1} \cup \hdots \cup \mathcal{N}_{p} \setminus \{i\}$ forms a new set of nearest neighbors, leading to a factor $\det(A(\mathcal{N}_{1} \cup \hdots \cup \mathcal{N}_{p} \setminus \{i\})) = \det(A(\mathcal{N}_{1} \cup \hdots \cup \mathcal{N}_{p})[i;i])$ instead of $\det(A(\mathcal{N}_{1})) \cdots \det(A(\mathcal{N}_{p}))$ in the final product \eqref{eq:determ2}. This proves expression \eqref{eq:Kin}.

Now consider a pair of points $i,j$ as in point \emph{ii)} of the theorem. If either $i$ or $j$ is the point closest in $\mathcal{N}$ to the root, we reorder the points in $\mathcal{P}$ as described above, but now with the extra constraint that within the set $\mathcal{N}$, all points are numbered consecutively. Note that it is always possible to find such an ordering for any one set $\mathcal{N}$ (but not for all sets at the same time). Assume furthermore that in the ordering of sets, $\mathcal{N}$ comes at the $k$'th place ($k \leq m$). Again, we start by factorizing out the determinant of $A(\mathcal{N}_m)$. If $k<m$ and if $l$, the point in $\mathcal{N}_m$ closest to the root, is not $j$, the factorization can be carried out as explained above by using the row operations $R_e(A) \rightarrow R_e(A) - a_{el} R_l(A)$ for $e \in \mathcal{N}_m$. If $l=j$, the deletion of the corresponding row prevents us from using it to execute these row operations, but the $j$'th column can be used instead: $C_e(A) \rightarrow C_e(A) - a_{le} C_l(A)$. Induction on $m$, until the set $\mathcal{N}$ is reached, then factors out the determinants of the matrices $A(\mathcal{N}_q)$, $q>k$. If $k=m$, the previous induction can be skipped. 

When $i$ nor $j$ are the point closest to the root in set $\mathcal{N}$, further factorization can proceed unhindered by using row operation, leading to a factorization similar to \eqref{eq:determ2}, except for the replacement $\det(A(\mathcal{N}_k)) \rightarrow \det(A(\mathcal{N}_k)[j;i])$.

When $i$ (resp. $j$) is the point closest to the root, our special numbering scheme allows the application of $R_d \rightarrow R_d - d_{di}R_i$ (resp. $C_d \rightarrow C_d - a_{id}C_i$) for $d \in \mathcal{N}_1 \cup \hdots \cup \mathcal{N}_{k-1}$, resulting in the matrix:
\begin{equation}
\det(A(\mathcal{N}_1 \cup \hdots \cup \mathcal{N}_{k})) = \det \left(
\begin{array}{cc}
B & 0 \\
A'' & A(\mathcal{N})[j;i]
\end{array} \right)
 = \det(B) \det(A(\mathcal{N})[j;i]),
\end{equation}
resp.
\begin{equation}
\det(A(\mathcal{N}_1 \cup \hdots \cup \mathcal{N}_{k})) = \det \left(
\begin{array}{cc}
B & A'' \\
0 & A(\mathcal{N})[j;i]
\end{array} \right)
 = \det(B) \det(A(\mathcal{N})[j;i]).
 \end{equation}
It can be checked that in both cases $\det(B) = \det(A(\mathcal{N}_1 \cup \hdots \cup \mathcal{N}_{k-1}))$. Further factorization can then proceed unhindered, leading again to \eqref{eq:determ2} with $\det(A(\mathcal{N}_k)) \rightarrow \det(A(\mathcal{N}_k)[j;i])$. This proves expression \eqref{eq:Ktrans}.

Finally, consider a pair of points $i>j$ that are not nearest neighbors. We assure that the points are numbered in such a way that all the sets $\mathcal{N}$ that contain $j$ have points with indices smaller than $i$ (note that since $i$ and $j$ are not nearest neighbors, none of these sets contains $i$). The familiar reduction scheme for $A[j;i]$ by applying row operations can be thought of as writing this matrix in `block-upper triangular' form. Its determinant is then the product of all determinants of the block-matrices on the diagonal. Between the $j$'th row and $i$'th column of this matrix, the diagonal is shifted by 1 to a lower row compared to the diagonal of $A$. Since $i$ does certainly not belong to the same `block' as $j$, there is at least one 0 on this diagonal (Fig~\ref{fig:formalism}B). This proves that $h_{ij}=0$. That $h_{ji}=0$, can be proved analogously, but now the matrix is written in `block-lower triangular' form by using column operations. As such, the statement \eqref{eq:stuff} is proven.
\end{proof}

\begin{corollary}\label{col:nok}
\textbf{Number of kernels.}
When a total of $m$ sets of nearest neighbors $\mathcal{N}_{1}, \hdots, \mathcal{N}_{m}$ exists within a set of $n$ points $\mathcal{P}$, the number of non-zero kernels $M$ in equation \eqref{eq:rewrite} is:
\begin{equation}\label{eq:nok}
M = n + \sum_{q=1}^{m} \abs{\mathcal{N}_q} \left(\abs{\mathcal{N}_q} - 1 \right),
\end{equation}
with $\abs{\mathcal{N}_q}$ the cardinality of the set $\mathcal{N}_q$.
\end{corollary}
\begin{proof} First, we prove equation \eqref{eq:nok}. For $n$ points, there are $n$ kernels $f_i$. In every set $\mathcal{N}_q$, there is kernels $h_{ij}$ between for every combination of points $i,j \in \mathcal{N}_q$, with $i \neq j$. Consequently, within a set $\mathcal{N}_q$, there are $\abs{\mathcal{N}_q} (\abs{\mathcal{N}_q} - 1)$ kernels. Hence, in total, there are $n+\sum_{q=1}^{m} \abs{\mathcal{N}_q} (\abs{\mathcal{N}_q} - 1)$ kernels.
\end{proof}

\begin{corollary}\label{col:sparseness}
\textbf{Sparseness.}
For a given tree graph and a given positive number $n$, the configurations of $n$ points that minimize the number $M$ of kernels in \eqref{eq:rewrite} have:
\begin{equation}
M = 3n-2.
\end{equation}
\end{corollary}
\begin{proof}
When $n=1$ the statement is trivial. For an optimal configuration with $n>1$, $\abs{\mathcal{N}_q} = 2$ for all $m$ sets $\mathcal{N}_q$, and $m=n-1$. Hence $M = n + 2(n-1) = 3n-2$.
\end{proof}

From the viewpoint of computational efficacy, corollaries \ref{col:nok} and \ref{col:sparseness} indicate that there are `well-chosen' and `badly-chosen' ways for the input locations to be distributed on the tree graph. For a `well-chosen' configuration $\abs{\mathcal{N}} = 2$, or at least $\abs{\mathcal{N}} \ll n$, for all $\mathcal{N}$, whereas in worst case scenarios there is only a single set $\mathcal{N}$. 

\subsection*{An efficient method to integrate the system of Volterra integral equations}
Transforming \eqref{eq:rewrite} to the time domain results in a system of Volterra integral equations:
\begin{equation}\label{eq:temprewrite}
V_i(t) = \int_{0}^{t} f_i(t-s)I_i(s, V(x_i,s))\mathrm{d}s + \sum_{j \neq i}  \int_{0}^{t}  h_{ij}(t-s)V_j(s)\mathrm{d}s \hspace{4mm} \forall i,
\end{equation}
where all kernels $h_{ij}$ between points $i, j$ that are not nearest neighbors are zero due to theorem \ref{thm:maintheorem}. In this form, the SGF-formalism is well-suited to simulate neuron models. 

Let $O(n_k)$ denote the number of operations required to compute a convolution each time-step. If these convolutions were to be integrated naively, by evaluating the quadrature explicitly (the Quad approach), $n_k$ would denote the number of time-steps after which the kernel can be truncated. However, the kernels are computed in the frequency domain, so if a partial fraction decomposition in this domain can be found:
\begin{equation}
\begin{aligned}
f_i(\omega) & \approx \sum_{l=1}^{L_{i}} \frac{\gamma_i^l}{i\omega +\alpha_i^l} \\
h_{ij}(\omega) & \approx \sum_{l=1}^{L_{ij}} \frac{\gamma_{ij}^l}{i\omega +\alpha_{ij}^l},
\end{aligned}
\end{equation}
where $\Re(\alpha) < 0$, then in the time domain the kernels can be readily expressed as sums of one sided exponentials:
\begin{equation}\label{eq:vectorfit}
\begin{aligned}
f_i(t) & \approx \sum_{l=1}^{L_{i}} \gamma_i^l e^{\alpha_i^l t} \\
h_{ij}(t) & \approx \sum_{l=1}^{L_{ij}} \gamma_{ij}^l e^{\alpha_{ij}^l t},
\end{aligned}
\end{equation}
when $t\geq 0$, and $f_i(t)=h_{ij}(t)=0$ otherwise. Such a decomposition can be derived accurately by employing the vector fitting (VF) algorithm \citep{Gustavsen1999} (see appendix \S\ref{app:VF}). Consequently, the convolutions with the kernels become sums of convolutions with exponentials, which can be readily expressed as simple differential equation (see for instance \citep{Lubich2002}). We call this the Exp approach and $n_k$ is the number of exponentials per kernel in this case. Usually, for the Exp approach, $n_k$ is smaller than for the Quad approach. Nevertheless, many of the exponentials of the VF algorithm are used to approximate the small $t$ behavior of the kernels, and hence have a very short time-scale. The large $t$ behaviour is often described by one or a few exponentials. This suggest that a `mixed' approach could yield optimal performance, where for small $t$ the quadrature is computed explicitly, and for large $t$ the convolution is computed as an ODE (or the sum of a few ODE's). We now give a detailed account of the mixed approach.

Let $t=Nh$, with $h$ the integration step and $N$ a natural number, and let us assume that $V_i(\tau)$ is known for all $i$ and for $\tau \in \{ t-kh; k=1,\hdots,N \}$, and that we want to know $V_i(t+h)$, with $h>0$. We split the convolutions in equation \eqref{eq:temprewrite} into an quadrature term ($\int_{t-Kh}^{t+h}$) and a exponential term ($\int_{0}^{t-Kh}$):
\begin{equation}\label{eq:convsplit}
\begin{aligned}
V_i(t+h) = & \int_{t-Kh}^{t+h} f_i(t+h-s)I_i(s)\mathrm{d}s + \sum_{j \neq i}  \int_{t-Kh}^{t+h}  h_{ij}(t+h-s)V_j(s)\mathrm{d}s + \\ 
& \int_{0}^{t-Kh} f_i(t+h-s)I_i(s)\mathrm{d}s + \sum_{j \neq i}  \int_{0}^{t-Kh}  h_{ij}(t+h-s)V_j(s)\mathrm{d}s,
\end{aligned}
\end{equation}
where $K$ is a natural number that needs to be chosen.
We assume that between the temporal grid points $t-kh$ and $t-(k+1)h$, $V_j(t)$ can be approximated by a linear interpolation:
\begin{equation}\label{eq:vlin}
V_j(s) \approx V_j(t-kh) + \frac{V_j(t-kh)-V_j(t-(k+1)h)}{h} (s-t+kh), \hspace{4mm} t-(k+1)h \leq s \leq t-kh.
\end{equation}
With this and the exponential approximation \eqref{eq:vectorfit}, the quadrature term for convolutions with $V_j$ becomes:
\begin{equation}
\begin{aligned}
\int_{t-Kh}^{t+h} h_{ij}(t+h-s) & V_j(s)\mathrm{d}s \approx \\ 
& V_j(t+h) \sum_l \left[-\frac{\gamma_{ij}^l}{\alpha_{ij}^l} + \frac{\gamma_{ij}^l}{\alpha_{ij}^{l \, 2} h} \left(e^{\alpha_{ij}^l h} - 1 \right) \right] \\ 
+ \sum_{k=0}^{K-1} & V_j(t-kh)  \sum_l \left[ - \frac{\gamma_{ij}^l}{\alpha_{ij}^{l \, 2} h} \left( e^{\alpha_{ij}^l (k+1) h} - 2 e^{\alpha_{ij}^l k h} + e^{\alpha_{ij}^l (k-1) h} \right) \right] \\
+ & V_j(t-Kh) \sum_l \left[ -\frac{\gamma_{ij}^l}{\alpha_{ij}^l} e^{\alpha_{ij}^l K h} +\frac{\gamma_{ij}^l}{\alpha_{ij}^{l \, 2} h} \left( e^{\alpha_{ij}^l K h} - e^{\alpha_{ij}^l (K-1) h} \right) \right],
\end{aligned}
\end{equation}
and similarly for the convolutions with the input $I_i$. For the exponential term, we get the following:
\begin{equation}\label{eq:expterm}
\int_{0}^{t-Kh} h_{ij}(t+h-s)V_j(s)\mathrm{d}s = \sum_l  \gamma_{ij}^l e^{\alpha_{ij}^l (K+1) h} \int_{0}^{t-hK} e^{\alpha_{ij}^l (t-Kh-s)} V_j(s) \mathrm{d}s,
\end{equation}
where $u_{ij}^l(t-Kh) \equiv \int_{0}^{t-Kh} \gamma_{ij}^l e^{\alpha_{ij}^l (t-Kh-s)} V_j(s) \mathrm{d}s$ is the solution of an initial value problem:
\begin{equation}
\begin{aligned}
& \dot{u}_{ij}^l(t) = \alpha_{ij}^l u_{ij}^l(t) + \gamma_{ij}^l V_j(t) \\
& u_{ij}^l(t=0) = 0,
\end{aligned}
\end{equation}
whose value can be computed recursively:
\begin{equation}
u_{ij}^l(t-Kh) = e^{\alpha_{ij}^l h} u_{ij}^l(t-(K+1)h) + \int_{t-(K+1)h}^{t-Kh} e^{\alpha_{ij}^l s} V_j(s) \mathrm{d}s.
\end{equation}
Using again the linear approximation \eqref{eq:vlin}, this becomes:
\begin{equation}
\begin{aligned}
u_{ij}^l(t-Kh) \approx e^{\alpha_{ij}^l h} & u_{ij}^l(t-(K+1)h) \\
+ & V_j(t-Kh) \left[-\frac{\gamma_{ij}^l}{\alpha_{ij}^l} + \frac{\gamma_{ij}^l}{\alpha_{ij}^{l \, 2} h} \left(e^{\alpha_{ij}^lh} - 1 \right) \right] \\ 
+ & V_j(t-(K+1)h) \left[\frac{\gamma_{ij}^l}{\alpha_{ij}^l} e^{\alpha_{ij}^l h} - \frac{\gamma_{ij}^l}{\alpha_{ij}^{l \, 2} h} \left(e^{\alpha_{ij}^lh} - 1 \right) \right].
\end{aligned}
\end{equation}
Analogous considerations apply for the convolutions with the input $I_i$. 

Let us now define a matrix $H_0$ by:
\begin{equation}\label{eq:matH}
(H_0)_{ij} = \sum_l \left[-\frac{\gamma_{ij}^l}{\alpha_{ij}^l} + \frac{\gamma_{ij}^l}{\alpha_{ij}^{l \, 2} h} \left(e^{\alpha_{ij}^l h} - 1 \right) \right] \\
\end{equation}
and a tensor $H_1$ by:
\begin{equation}\label{eq:tenH}
\begin{aligned}
(H_1)_{ij}^k & = \sum_l \left[ - \frac{\gamma_{ij}^l}{\alpha_{ij}^{l \, 2} h} \left( e^{\alpha_{ij}^l (k+1) h} - 2 e^{\alpha_{ij}^l k h} + e^{\alpha_{ij}^l (k-1) h} \right) \right] \hspace{4mm} \text{ for } k = 0,\hdots, K-1 \\
(H_1)_{ij}^k & = \sum_l \left[ -\frac{\gamma_{ij}^l}{\alpha_{ij}^l} e^{\alpha_{ij}^l k h} +\frac{\gamma_{ij}^l}{\alpha_{ij}^{l \, 2} h} \left( e^{\alpha_{ij}^l k h} - e^{\alpha_{ij}^l (K-1) h} \right) \right] \hspace{4mm} \text{ for } k = K. 
\end{aligned}
\end{equation}
when $i,j$ are nearest neighbors, and $(H_0)_{ij}=(H_1)_{ij}^k=0$ otherwise. For the input kernels $f_i$, we define a vector $F_0$ (with $(F_0)_i$ an analogous to \eqref{eq:matH}) and a matrix $F_1$ (with $(F_1)_i^k$ analogous to \eqref{eq:tenH}).
Using these definitions, we may write equation \eqref{eq:convsplit} as:
\begin{equation}\label{eq:convsplitrewrite}
\begin{aligned}
V_i(t+h) = & \left[ (F_0)_i I_i(t+h) + \sum_{k=0}^{K} (F_1)_i^k I_i(t-kh) \right] \\ & +  \sum_{j \neq i} \left[ (H_0)_{ij} V_j(t+h) + \sum_{k=0}^{K}(H_1)_{ij}^k V_j(t-kh) \right] \\ 
& + \sum_{l} e^{\alpha_i^l (K+1) h} u_i^l(t-Kh) + \sum_{j \neq i} \sum_{l} e^{\alpha_{ij}^l (K+1) h} u_{ij}^l(t-Kh).
\end{aligned}
\end{equation}
To simplify the notation, we group all terms that do not contain voltage or input at time $t+h$ in a vector $\mathbf{k}(t)$, for which:
\begin{equation}\label{eq:kvect}
\begin{aligned}
k_i(t) = & \left[\sum_{k=0}^{K} (F_1)_i^k I_i(t-kh) \right] + \sum_{j \neq i} \left[ \sum_{k=0}^{K}(H_1)_{ij}^k V_j(t-kh) \right] \\ 
& + \sum_{l} e^{\alpha_i^l (K+1) h} u_i^l(t-Kh) + \sum_{j \neq i} \sum_{l} e^{\alpha_{ij}^l (K+1) h} u_{ij}^l(t-Kh)
\end{aligned}
\end{equation}
Consequently, equation \eqref{eq:convsplitrewrite} becomes in matrix form:
\begin{equation} \label{eq:solve}
(\mathbb{I} - H_0) \mathbf{V}(t+h) = \text{diag}(\mathbf{F_0}) \mathbf{I}(t+h) + \mathbf{k}(t).
\end{equation}
Solving this matrix equation then gives $ \mathbf{V}(t+h)$.

Note that when the input current $I_i$ depends on the local voltage, equation \eqref{eq:solve} is only semi-implicit, since the voltage at time $t$ would be required to compute the current at time $t+h$. This description may be unstable for certain ion channels. However, the direct dependence of most currents on the voltage is linear (in the case of ion-channels for instance, non-linearities arise through the non-linear dependence of state variables on the voltage):
\begin{equation}\label{eq:Iexpr}
I_i(t, V(x_i,t)) = c_i(t) + d_i(t) V(x_i,t).
\end{equation}
Using this, equation \eqref{eq:solve} can be modified in the following way:
\begin{equation}\label{eq:extsolve}
(\mathbb{I} - H_0-\text{diag}(\mathbf{F_0} \odot \mathbf{d}(t+h)) \mathbf{V}(t+h) = \text{diag}(\mathbf{F_0}) \mathbf{c}(t+h) + \mathbf{k}(t).
\end{equation}
where $\odot$ denotes the element wise multiplication. This way of integration, while still being semi-implicit (since the equations for state variables of ion channels \eqref{eq:chan} are still integrated explicitly) is stable and standard in neuroscientific applications \citep{Hines1984}.

Note furthermore that the matrix $(\mathbb{I} - H_0) $ resp. $\mathbb{I} - H_0-\text{diag}(\mathbf{F_0} \odot \mathbf{d}(t+h)$ is structurally symmetric: whenever an off-diagonal element is nonzero, its counterpart opposite to the diagonal is nonzero as well. When $\abs{\mathcal{N}} = 2$ for all $\mathcal{N}$, and when the input locations are ordered in the right way, this matrix is a Hines matrix \citep{Hines1984}. For such a matrix a linear system of equations of the form \eqref{eq:solve} resp. \eqref{eq:extsolve} can be solved for $\mathbf{V}$ in $O(n)$ steps instead of the usual $O(n^3)$ steps.

Note also that we did not discuss the initialization steps of this algorithm explicitly. We omitted this discussion since, in neuroscience, it usually is simply assumed that the neuron model is at equilibrium at all times $t<0$. Hence, the vectors $\mathbf{k}(t=0), \mathbf{k}(t=h), \mathbf{k}(t=2h), \hdots $ can be computed easily by assuming that $V(-kh) = 0$ for all required $k$'s.

Finally, we remark that $K$, the parameter in the algorithm which determines the limit $Kh$ below which the quadrature is evaluated explicitly, can be chosen to minimize the number $n_k$ of operations per kernel. In principle, $K$ could be chosen separately for each kernel. In the present derivation we however opted not to do so for simplicity.

\subsection*{Nonlinear terms and the small $\Delta x$ limit of the SGF formalism}
While a set of ion channels that behaves approximately linear may be incorporated directly in the SGF formalism, the question what to do with channels that act truly non-linear remains. Such channels can be moved from the left hand side of equation \eqref{eq:cable} to the right hand side, and treated as an input current that depends on the local voltage. This current needs to be `compartmentalized', i.e. expressed at a discrete number of input locations, with a certain separation $\Delta x$.

A `recompartmentalization' of this type leads to the question of what the relation between the SGF formalism and the canonical second order finite difference approximation is. In both approaches, the input is compartmentalized. In the finite difference approximation however, the entire cable (i.e. the longitudinal, capacitive and leak currents) is compartmentalized as well, whereas in the SGF formalism, the cable is treated exactly. Consequently, in spatial regions with truly non-linear ion channels, the accuracy of the SGF is equivalent to (or better than, since the linear currents are still treated exactly) the accuracy of the canonical second order finite difference approximation. 

The observation that the SGF formalism treats the cable exactly then leads to the question whether the second order finite difference approximation can be derived from it, as both approaches describe the voltage at a given location only as a function of the voltage at the neighboring locations and the input at that location. We show in a simplified setting that, when the distance $\Delta x$ between the input locations is sufficiently small, the formalism reduces to the second order finite difference approximation of the original equation. 

Although we expect that the reduction of the SGF formalism to the second order finite difference approximation is valid for all equations of the type \eqref{eq:PDE} and for arbitrary tree graphs, its derivation can become prohibitively complex. We therefore constrain ourselves to the passive cable equation \eqref{eq:cable} and to the case where this equation is defined on a line of length 1, with a homogeneous boundary condition at each end. We suppose that there are $n-1$ input locations, distributed evenly with spacing $\Delta x = \slfrac{1}{n}$.
\begin{equation}\label{eq:simpcable}
\begin{cases}
\frac{\partial V}{\partial t}(x,t) + \frac{\partial^2 V}{\partial x^2}(x,t) - V(x,t) = \sum_{i=1}^{n-1} I_i(t) \delta(x-i\Delta x). \\
\frac{\partial V}{\partial x}(0,t) = \hat{L}_A V(0,t) \\
\frac{\partial V}{\partial x}(1,t) = \hat{L}_B V(1,t) 
\end{cases}
\end{equation}
Note that the input currents $I_i(t)$ can depend on the local voltage $V(x_i,t)$, as they can be the result of a discretization of the ion channel currents.

\paragraph{Second order finite difference approximation}
To obtain the second order finite difference approximation of equation \eqref{eq:simpcable}, we replace
\begin{equation}
\frac{\partial^2 V}{\partial x^2}(x_i,t) \rightarrow \frac{V_{i+1}(t) - 2V_i(t) + V_{i-1}(t)}{\Delta x^2}
\end{equation}
and we average the inputs for each compartment:
\begin{equation}
\frac{1}{\Delta x} \int_{(i-\slfrac{1}{2})\Delta x}^{(i+\slfrac{1}{2})\Delta x}  I_i(t) \delta(x-i\Delta x) \mathrm{d}x = \frac{I_i(t)}{\Delta x}
\end{equation}
Consequently, we find for $i=1,\hdots,n-1$:
\begin{equation}
\frac{V_{i+1}(t) - 2V_i(t) + V_{i-1}(t)}{\Delta x^2} + \dot{V}_i(t) - V_i(t) = \frac{I_i(t)}{\Delta x}.
\end{equation}
The boundary condition at $x=0$ becomes:
\begin{equation}
\frac{V_1(t)-V_0(t)}{\Delta x} =  \hat{L}_A^{(t)} V_0(t),
\end{equation}
and an analogous expression applies for the boundary condition at $x=1$

\paragraph{Reduction of the SGF formalism}
The Fourier transform of the problem \eqref{eq:simpcable} reads:
\begin{equation}\label{eq:cablefourier}
\begin{cases}
\frac{\partial^2 \fV}{\partial x^2}(x,\omega) - \gamma(\omega)^2 \fV(\omega) = \sum_{i=1}^n \fI_i(\omega) \delta(x-x_i) \\
\frac{\partial \fV}{\partial x}(0,\omega) = \tilde{L}_A(\omega)\fV(0,\omega) \\
\frac{\partial \fV}{\partial x}(1,\omega) = \tilde{L}_B(\omega)\fV(1,\omega) ,
\end{cases}
\end{equation}
where $\tilde{L}_A$ and $\tilde{L}_B$ are polynomials in $\omega$ representing the respective Fourier transforms of $\hat{L}_A$ and $\hat{L}_B$, and where $\gamma^2(\omega) = i\omega - 1$. In the following the coordinate $\omega$ is dropped for notational clarity. We define two linearly independent solutions to the homogeneous problem:
\begin{equation}\label{eq:sol}
\begin{aligned}
u_A(x) = e^{\gamma x} + k_A e^{-\gamma x} \\
u_B(x) = e^{\gamma x} + k_B e^{-\gamma x}
\end{aligned}
\end{equation}
that satisfy the boundary conditions respectively in $x=0$ ($u_A$) and $x=1$ ($u_B$). Then the Green's function of problem \eqref{eq:cablefourier} is given by \citep[page 66]{Stakgold1967}:
\begin{equation}\label{eq:gf}
g_{ij} = 
\begin{cases}
\frac{u_A(x_i)u_B(x_j)}{-2 \gamma (k_B - k_A)} \hspace{4mm} \text{ if } x_i \leq x_j\\
\frac{u_B(x_i)u_A(x_j)}{-2 \gamma (k_B - k_A)} \hspace{4mm} \text{ if } x_i \geq x_j.
\end{cases}
\end{equation}
Consequently, the attenuation functions are given by:
\begin{equation}\label{eq:att}
a_{ij} = 
\begin{cases}
\frac{u_A(x_i)}{u_A(x_j)} \hspace{4mm} \text{ if } x_i \leq x_j\\
\frac{u_B(x_i)}{u_B(x_j)} \hspace{4mm} \text{ if } x_i \geq x_j.
\end{cases}
\end{equation}
According to theorem \ref{thm:maintheorem}, problem \eqref{eq:cablefourier} is then fully defined by the following kernels:
\begin{itemize}
\item for $i=0,\hdots,n-2$
\begin{equation}
h_{i+1,i} = \frac{a_{i+1,i} - a_{i+1,i+2}a_{i+2,i}}{1 - a_{i, i+2}a_{i+2,i}} = a_{i+1,i} \frac{1 - a_{i+1,i+2}a_{i+2,i+1}}{1 - a_{i, i+2}a_{i+2,i}}
\end{equation}
\item for $i=2, \hdots, n-1$
\begin{equation}
h_{i,i+1} = \frac{a_{i,i+1} - a_{i,i-1}a_{i-1,i+1}}{1 - a_{i+1, i-1}a_{i-1,i+1}} = a_{i,i+1} \frac{1 - a_{i,i-1}a_{i-1,i}}{1 - a_{i+1, i-1}a_{i-1,i+1}}
\end{equation}
\item $h_{01} = a_{01}$
\item $h_{n,n-1} = a_{n,n-1}$
\item for $i=1,\hdots,n-1$
\begin{equation}
f_i = g_{ii} \frac{(1-a_{i,i-1}a_{i-1,i}) (1-a_{i+1,i}a_{i,i+1})}{1-a_{i+1,i-1}a_{i-1,i+1}}
\end{equation}
\end{itemize}
Using equations \eqref{eq:att} and \eqref{eq:sol}, it can be checked that:
\begin{equation}
1-a_{ij}a_{ji} = 
\begin{cases}
\frac{2 (k_B-k_A) \sinh( \gamma (x_j-x_i))}{u_B(x_i)u_A(x_j)} \hspace{4mm} \text{ if } x_i \leq x_j \\
\frac{2 (k_B-k_A) \sinh( \gamma (x_i-x_j))}{u_A(x_i)u_B(x_j)} \hspace{4mm} \text{ if } x_i \geq x_j.
\end{cases}
\end{equation}
Consequently, it follows that:
\begin{itemize}
\item for $i=0,\hdots,n-2$
\begin{equation}
h_{i+1,i} = \frac{\sinh(\gamma \Delta x)}{\sinh (2 \gamma \Delta x)}
\end{equation}
\item for $i=2, \hdots, n-1$
\begin{equation}
h_{i,i+1} = \frac{\sinh(\gamma \Delta x)}{\sinh (2 \gamma \Delta x)}
\end{equation}
\item for $i=1,\hdots,n-1$
\begin{equation}
f_i = -\frac{1}{\gamma} \frac{(\sinh(\gamma \Delta x))^2}{\sinh (2 \gamma \Delta x)}
\end{equation}
\end{itemize}
With these equations, equation \eqref{eq:rewrite} for $i=1,\hdots,n-1$ becomes:
\begin{equation}
\fV_i =  -\frac{1}{ \gamma} \frac{(\sinh(\gamma \Delta x))^2}{\sinh (2 \gamma \Delta x)} \fI_i + \frac{\sinh(\gamma \Delta x)}{\sinh (2 \gamma \Delta x)} \left( \fV_{i-1} + \fV_{i+1} \right).
\end{equation}
Since $\Delta x$ is small, the $\sinh$-functions can be expanded. This gives:
\begin{equation}
\fV_i =  - \frac{\Delta x}{2 + \frac{8}{6} \gamma^2 \Delta x^2} \fI_i + \frac{1 + \frac{1}{6} \gamma^2 \Delta x^2}{2 + \frac{8}{6} \gamma^2 \Delta x^2} \left( \fV_{i-1} + \fV_{i+1} \right).
\end{equation}
Multiplying both sides by the denominator $2 + \frac{8}{6} \gamma^2 \Delta x^2$, rearranging terms and dividing by $\Delta x^2$ then gives:
 \begin{equation}
\frac{\fV_{i-1} - 2\fV_i + \fV_{i+1}}{\Delta x^2} -  \frac{8}{6} \gamma^2 \fV_i + \frac{1}{6} \gamma^2 \left( \fV_{i-1} + \fV_{i+1} \right) =  \frac{1}{\Delta x} \fI_i.
\end{equation}
Averaging $\fV_{i-1} + \fV_{i+1} \approx 2 \fV_i$ and transforming the resulting equation back to the time domain then results precisely in the finite difference approximation for $i=1,\hdots,n-1$.

Let us now investigate equation \eqref{eq:rewrite} when $i=0$. Here, it holds that:
\begin{equation}
\fV_0 = a_{01}\fV_1.
\end{equation}
Substituting the explicit form of $a_{01} = \slfrac{u_A(0)}{u_A(\Delta x)}$ in this equation and expanding the exponentials up to first order leads to:
\begin{equation}\label{eq:V0}
\fV_0 = \frac{1}{1 + \frac{1-k_A}{1+k_A} \gamma \Delta x} \fV_1
\end{equation}
Since $u_A$ satisfies the boundary condition at $x=0$, it can be checked that $\gamma \slfrac{(1-k_A)}{(1+k_A)} = \tilde{L}_A$, so that equation \eqref{eq:V0} can be rewritten as:
\begin{equation}
\tilde{L}_A \fV_0 = \frac{\fV_1-\fV_0}{\Delta x}.
\end{equation}
Transforming this equation back to the time domain yields precisely the finite difference approximation for $V_0$. Analogous considerations apply for $V_n$.

\subsection*{Implementation}
In this section we summarize the steps one has to undertake to implement the SGF formalism. Such an implementation consists of two main parts: an initialization part and a simulation part. 

\noindent The initialization part consists of the following steps:
\begin{enumerate}
\item For any given set of input locations and a tree graph, a routine has to be implemented that can identify the sets of nearest neighbors. Such a routine could for instance return a list, where each element is a list in itself, containing the input locations that constitute a set of nearest neighbors.
\item The GF's $g_{ij}(\omega)$ between every pair of elements in the sets of nearest neighbors have to be computed. We implemented the algorithm of \citep{Koch1985} for this purpose. Note that for a set $\mathcal{N}$ of cardinality $\abs{\mathcal{N}}$, only $\slfrac{\abs{\mathcal{N}}(\abs{\mathcal{N}}+1)}{2}$ function $g_{ij}$ have to be computed, since $g_{ij}=g_{ji}$ \citep{Koch1998}. From these kernels, the attenuation functions can be derived easily.
\item The kernels $f_i$ and $h_{ij}$ have to be computed. This can be done either by evaluating the formulas \eqref{eq:Kin} and \eqref{eq:Ktrans} explicitly (when all matrices of which the determinants have to be computed are small) or by the derived formulas:
\begin{equation}
\begin{aligned}
f_i & = g_{ii}\frac{1}{A(\mathcal{N}_{1} \cup \hdots \cup \mathcal{N}_{p})^{-1}} \\
h_{ij}  & = - \frac{A(\mathcal{N}_{1} \cup \hdots \cup \mathcal{N}_{p} \cup \mathcal{N})_{ij}^{-1}}{A(\mathcal{N}_{1} \cup \hdots \cup \mathcal{N}_{p} \cup \mathcal{N})_{ii}^{-1}}.
\end{aligned}
\end{equation}
Note that in both cases, attenuation function are needed that are not directly computed in the previous step, since attenuation matrices of unions of sets are considered. This is not a problem, as these attenuation functions can be easily reconstructed from the transitivity property.
\item The partial fraction decomposition of each kernel has to be computed. The VF algorithm \citep{Gustavsen1999} (appendix \S\ref{app:VF}) is well suited for this purpose. When the parameters of the partial fraction decomposition are known, the vector $F_0$, the matrices $F_1$ and $H_0$ and the tensor $H_1$ can be computed, as described in \eqref{eq:matH} and \eqref{eq:tenH}. Since $H_0$ resp. $H_1$ are sparse in their indices $i$ and $j$, they can be stored as index-number resp. index-array pairs.
\end{enumerate}
\noindent For the simulation part, two main routines have to be called each time step:
\begin{enumerate}
\item A routine that computes the vectors $\mathbf{c}(t)$ and $\mathbf{d}(t)$ \eqref{eq:Iexpr}. This routine also advances synaptic conductances and ion channel state variables and its details are well established in the neuroscientific literature (see \citep{Rotter1999} and \citep{Hines1984}).
\item A routine that first computes the vector $\mathbf{k}(t)$ \eqref{eq:kvect}  and then solves equation \eqref{eq:extsolve}.
\end{enumerate}
We implemented a prototype of the above outlined initialization algorithm in Python. Tree structures were implemented using the btmorph library \citep{Torben-nielsen2014a}. The simulation algorithm was implemented both in pure Python and in C++ with a Cython interface.

\section{Validation}

\begin{figure}
   \centering
   \includegraphics[width=\textwidth]{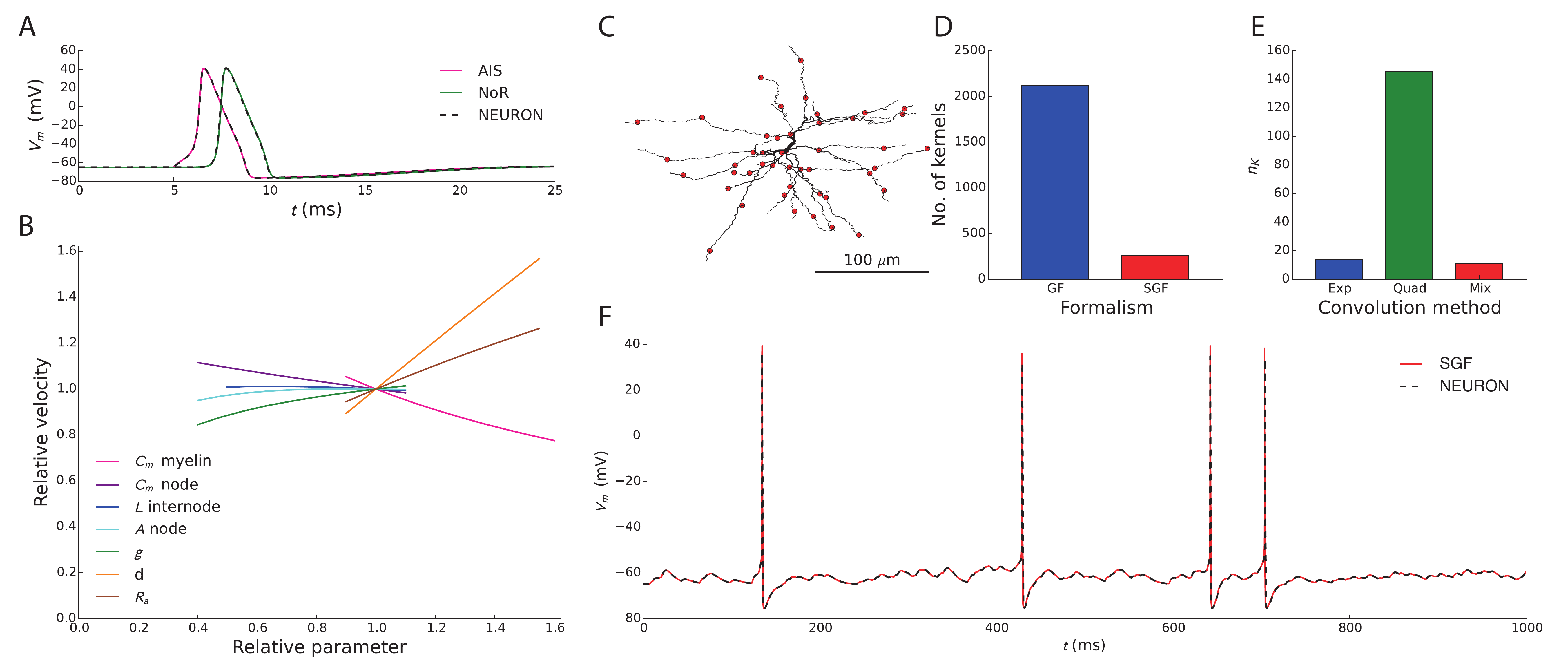}
   \caption{\textbf{Validation of the SGF formalism on axon (A-B) and dendrite models (C-F).} We implemented the axon model described in \citep{Moore1978} in the SGF formalism. A: validation of the action potential (AP) shape by comparison with an equivalent \textsc{neuron} model at the axon initial segment (AIS) and the 10'th node of Ranvier (NoR). B: reproduction of Fig~2 in \citep{Moore1978}, where the dependence of AP velocity on different biophysical parameters is studied. C: The dendritic morphology together with 50 synaptic input locations (the morphology was retrieved from the NeuroMorpho.org repository \citep{Ascoli2006} and originally published in \citep{Wang2002}). D: The number of kernels in the SGF formalism, compared to the number of kernels that would have been required in the normal GF formalism. E: The number of operation required per kernel to achieve similar levels of accuracy for 3 approaches: Exp, where all the exponentials from the VF algorithm are integrated, Quad, where the quadrature is computed explicitly and Mix, where we compute the quadrature for the first $K$ (here 3) steps, and use the ODEs to compute the rest of the convolution. F: Voltage trace at the soma upon stimulation of the synapses with 5 Hz Poisson spike trains.}
   \label{fig:validation}
\end{figure}

In myelinated axons, the approximation of grouping active currents at a discrete set of input locations holds exactly, as the only spots where active currents are present are the nodes of Ranvier. These nodes are separated by stretches of myelinated fiber of up to 2 mm in length, which can be modeled by equation \eqref{eq:cable}. We reproduce the model of \citep{Moore1978}, where the soma, axon initial segment and nodes of Ranvier are equipped with Hodgkin-Huxley \citep{HodgkinA.L.andHuxley1990} channels, so that the input current is of the form:
\begin{equation}\label{eq:HH}
I_i(t) = - \bar{g}_{Na} m_i(t)^3 h_i(t) (V(x_i,t)-E_{Na}) + \bar{g}_K n_i(t)^4 (V(x_i,t)-E_K) - \bar{g}_L (V(x_i,t)-E_L).
\end{equation}
Consequently, $c_i(t)$ and $d_i(t)$ in equation \eqref{eq:Iexpr} are given by:
\begin{equation}\label{eq:substitute}
\begin{aligned}
c_i(t) & =  \bar{g}_{Na} m_i(t)^3 h_i(t) E_{Na} + \bar{g}_K n_i(t)^4 E_K + \bar{g}_L E_L \\
d_i(t) & = - \bar{g}_{Na} m_i(t)^3 h_i(t) - \bar{g}_K n_i(t)^4 - \bar{g}_L,
\end{aligned}
\end{equation}
and the variables $m_i(t), h_i(t)$ and $n_i(t)$ evolve according to the following equations:
\begin{equation}
\begin{aligned}
\dot{m}_i(t) & = \frac{m_{\text{inf}}(V(x_i,t)) - m_i(t)}{\tau_m(V(x_i,t))} \\
\dot{h}_i(t) & = \frac{h_{\text{inf}}(V(x_i,t)) - h_i(t)}{\tau_h(V(x_i,t))} \\ 
\dot{n}_i(t) & = \frac{n_{\text{inf}}(V(x_i,t)) - n_i(t)}{\tau_n(V(x_i,t))},
\end{aligned}
\end{equation}
where $m_{\text{inf}}, h_{\text{inf}}, n_{\text{inf}}, \tau_m, \tau_h$ and $\tau_n$ are functions of the local voltage.

We implemented this model using the SGF formalism. In Fig~\ref{fig:validation}A, we validate the action potential (AP) shape at the axon initial segment and at the 10'th node of Ranvier by comparison with the \textsc{neuron} simulator \citep{Carnevale2006} -- the gold standard in neuronal modeling. In Fig~\ref{fig:validation}B we reproduce Fig~2 in \citep{Moore1978}, where the dependence of the AP velocity on various axonal parameters is investigated.

Finally, to show that the SGF formalism can also handle complex trees, we equipped a stellate cell morphology, comprising an active soma and passive dendrites (Fig~\ref{fig:validation}C) (modeled by equation \eqref{eq:cable}), with 50 conductance based synapses, so that
\begin{equation}
\begin{aligned}
I_i(t) & = \bar{g} (A_i(t) - B_i(t)) (V(x_i,t) - E_r) \\
\dot{A}_i(t) & = \frac{-A_i(t) + c_A \sum_{j} \delta(t-s_i^{(j)})}{\tau_A} \\
\dot{B}_i(t) & = \frac{-B_i(t) + c_B \sum_{j} \delta(t-s_i^{(j)})}{\tau_B},
\end{aligned}
\end{equation}
where $s_i^{(j)}$ represents the $j$'th spike time of the input spike train arriving at the $i$'th synapse. In this equation, the variables $A$ and $B$ are used to generate a double exponential profile (see for instance \citep{Rotter1999}). In this model, it holds that:
\begin{equation}
\begin{aligned}
c_i(t) & = - \bar{g} (A_i(t) - B_i(t)) E_r \\
d_i(t) & = \bar{g} (A_i(t) - B_i(t)).
\end{aligned}
\end{equation}
Excellent agreement with the \textsc{neuron} simulator is achieved (Fig~\ref{fig:validation}F). The amount of kernels required is far lower than in the normal GF formalism (Fig~\ref{fig:validation}D) and optimal performance (for which we found $n_k \approx 7$) was achieved for $K=3$ (Fig~\ref{fig:validation}E).

\subsection*{Computational cost}
\begin{figure}
   \centering
   \includegraphics[width=.7\textwidth]{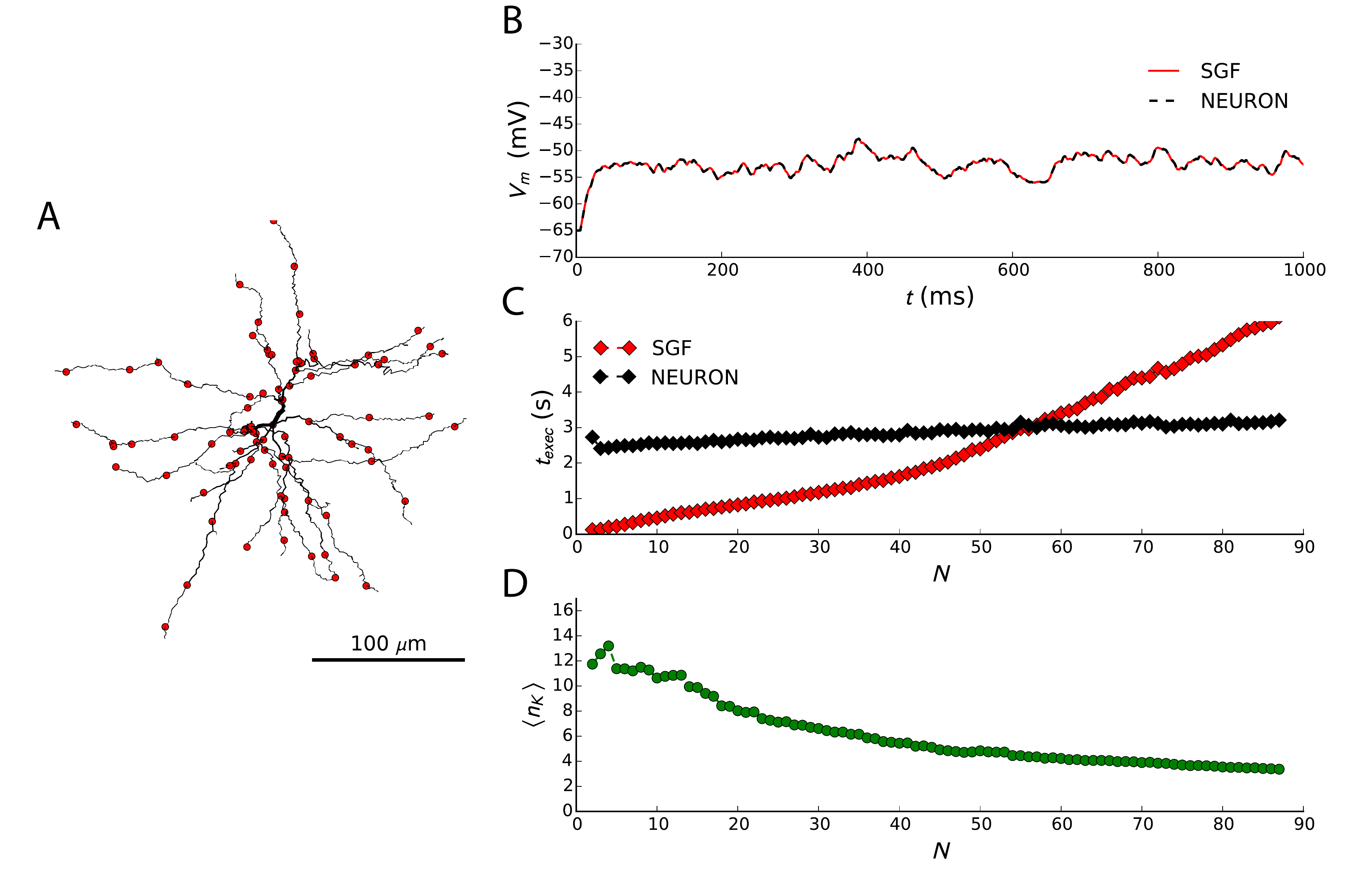}
   \caption{\textbf{The execution time of the SGF formalism compared to a \textsc{neuron} simulation.} A: The setup, showing the morphology and the maximal set of input locations $N=74$ input location. Note that this set was chosen so that $\abs{\mathcal{N}} = 2$ for all $\mathcal{N}$. B: The somatic voltage trace for the first second of a 10 s simulation, with a conductance based excitatory synapse at each input location indicated in A. C: Execution times as a function of the number of input locations. D: The average number of operations per kernel as a function of the number of input locations.}
   \label{fig:runtime}
\end{figure}

It is immediately clear that the computational cost of the SGF formalism is far lower than the computational cost of the GF formalism \citep{Wybo2013} from the number of required kernels alone (Fig~\ref{fig:validation}D). However, the comparison of computational cost with the canonical finite difference approaches requires a more careful discussion. For $n$ input locations and $n_t$ time steps, an explicit solver would typically require $O(n_k n_t n)$ steps, whereas a canonical finite difference approach would require $O(d_t n_t n_x)$ steps, with $n_x$ the number of spatial locations at which the finite differences have to be evaluated and $d_t$ the maximal degree of temporal differentiation in the operators $\hat{L}_i^{(t)}(x)$. Thus, we expect the SGF formalism to improve performance when $c n_k n < d_t n_x$, where $c$ is an a-priori unknown constant that may depend on the specific implementation. When the system is stiff, an implicit solver is typically needed to guarantee stability. Such solvers require a matrix of size $n$ (resp. $n_x$ in the case of finite differences) to be inverted each time step, normally an operation of complexity $O(n_x^3)$. However, for second order finite difference solvers, the tree graph introduces a special structure in this matrix \citep{Hines1984}, so that it can be inverted in $O(n_x)$ steps. In the SGF formalism, this structure can be maintained if $\abs{\mathcal{N}} = 2$ for all $\mathcal{N}$ and in this case the same performance criterion $c n_k n < d_t n_x$ applies as for explicit solvers. Note that for higher order finite difference approaches the matrix inversion still requires $O(n_x^3)$.

To test whether this constant $c$ is not excessively large, so that it may impede the usefulness of the SGF formalism, we compared the execution time of the SGF formalism with the execution time of the \textsc{neuron} simulator \citep{Carnevale2006}. In order to obtain a `fair' comparison, attention must be given to the exact type of implementation. For instance, comparing an explicit simulation paradigm with \textsc{neuron} makes little sense, as \textsc{neuron} implements a semi-implicit paradigm and would be severely at the disadvantage. On the other hand, \textsc{neuron} makes use of the Hines algorithm to invert the systems' matrix \citep{Hines1984}, and distributing input location completely at random in the SGF formalism (and using a semi-implicit paradigm as well) would exclude the use of this algorithm, which would severely disadvantage the SGF formalism. We therefor opted to implement a simulator in C++, that implements the same semi-implicit method described in equations \eqref{eq:extsolve} and \eqref{eq:solve}, but where the inversion uses the Hines algorithm. This restricts the input locations to configurations where $\abs{\mathcal{N}} = 2$ for all $\mathcal{N}$. We distributed 2 to 74 input locations in such a configuration (Fig~\ref{fig:runtime}A shows this configuration for $n=74$) on the same stellate cell morphology used in Fig~\ref{fig:validation} (panels C-E). No active channels where added to the model. In \textsc{neuron}, the cell model contained 299 compartments for a total dendritic length of approximately $4000$ $\mu$m ($\Delta x \approx 13.5$ $\mu$m). We added one excitatory, conductance based synapse at each input location and gave it a Poisson spike train of a certain rate, so that the total presynaptic rate at all synapses together was 1000 Hz. In Fig~\ref{fig:runtime}B we show the first second of a 10 s simulation where $n=74$. It can be seen that both traces agree impeccably. We then compared the execution times of both models for an increasing number of input locations (Fig~\ref{fig:runtime}C) for a simulation time of 10 s at a time step of 0.1 ms. It can be seen that the SGF implementation outperforms the \textsc{neuron} model until $n \approx 55$ (on a machine equipped with an Intel Core i7-3770 CPU @ 3.4GHz and 16 GB of RAM, running ubuntu 14.04 LTS). In Fig~\ref{fig:runtime}D, we plotted $\langle n_k \rangle$, the average number of operations per kernel and per time step. It can be seen that this number correlates with the execution time. When there are few input locations, kernels are in general longer and hence this number is higher. Then, when more kernels are added, this number drops quite rapidly. This explains the small slope of the execution times in the SGF formalism until $n \approx 35$, associated with the rapid decrease of $\langle n_k \rangle$. After that, $\langle n_k \rangle$ decreases slower, and hence the slope of the execution times becomes steeper.

\section{Summarizing remarks}
We have proven that on tree graphs, the GF formalism for linear, non-homogeneous, time-invariant PDE's with inputs at a discrete set of $n$ well-chosen locations requires only $O(n)$ rather than $n^2$ kernels, and termed this the SGF formalism. We discussed the meaning of `well-chosen', namely that the sizes of the sets of nearest neighbors must be small (ideally $\abs{\mathcal{N}} = 2$ for all $\mathcal{N}$). We have shown furthermore for equation \eqref{eq:cable} that in the limit of a small distances between the input locations, the SGF formalism can be reduced to the second order finite difference approximation. Thus, in some sense, the SGF formalism can be seen as a generalization of the second order finite difference approximation to arbitrary distance steps (as long as what lies between the input locations is approximately linear). We also employed the VF algorithm \citep{Gustavsen1999}, that fits the kernels with sums of exponentials, to design an efficient simulation algorithm. 

We then validated our algorithm on two neuroscientific problems: the modeling of myelinated axons and of dendritic integration. We showed that there was excellent agreement between the SGF formalism and the de facto standard \textsc{neuron} simulator. We found furthermore that, even for complex, branched morphologies, excellent accuracy was achieved for $n_k \lesssim 15$. We also discussed when an increase in computational performance is expected from the SGF formalism, and furthered our findings by comparing an efficient C++ implementation of the SGF formalism with the same model in \textsc{neuron}. Finally, we discuss in a broader sense when the SGF formalism is expected to provide advantage over other approaches.

\section{Discussion}

It can be seen that the question whether one should prefer the SGF formalism over second order finite difference approaches depends on the values of $n_k$ and $n$, which are not necessarily independent. Fig~\ref{fig:runtime} illustrates this: for the same dendritic arborization, and using the passive cable equation ($d_t=1$), $n_k$ ranges from $\approx 14$ for $n=2$ to $\approx 4$ for $n=74$. This leads to $n_k n \approx 28$ to $296$ operations per time step. The dendritic arborization on the other hand has a total length of approximately $4000$ $\mu$m. According to the Lambda rule \citep{Carnevale2006}, the  canonically accepted method of calculating the distance step between compartments,  $\Delta x \approx 10-15$ $\mu$m. Consequently, a second order finite difference approximation would use approximately $n_x \approx 300$. Hence this estimate would indicate that performance can be gained when $n \lesssim 75$. That we found this number to be slightly lower ($n \lesssim 55$) in the simulations we conducted (Fig~\ref{fig:runtime}) may be due to implementation and optimization related factors. One important remark is that for efficient semi-implicit simulations, the input locations must be chosen so that $\abs{\mathcal{N}} = 2$ for all $\mathcal{N}$ to be able to use the Hines algorithm. The question may now be asked whether this is overly restrictive. We believe that this is not the case: placing an input location at the soma for instance separates all main dendrites. Bifurcations in the dendrites may induce further sets of nearest neighbors of sizes slightly larger than 2. Whenever this occurs however, an input location can be added at the bifurcation point, so that such a set is split as well.

Thus the choice of the SGF formalism over the finite difference approximation depends on the neural system at hand. In some systems, such as bipolar neurons used in auditory coincidence detection \citep{Agmon-Snir1998, Wybo2013}, inputs occur only at a small set of locations, and the cells' computation is performed by linear membrane properties. Other systems, such as myelinated axons, have non-linear membrane currents that are concentrated at a discrete set of locations -- the nodes of Ranvier, with stretches of myelinated fiber in between that behave approximately linear. It is to be expected that in such systems, the SGF formalism may significantly improve performance over the second order finite difference approach. In other systems, such as cortical cells, inputs are distributed throughout the dendritic arborization in an almost continuous fashion. In these cells, the answer to the question whether the SGF formalism yields computational advantage is negative, when one aims at retaining all biophysical detail. Indeed, it is clear that when the number of input locations in the SGF formalism approaches the number of compartments \textsc{neuron} requires, there is no computational gain in using the former. Nevertheless, computational neuroscientists have been seeking ways of reducing the cost of simulating these cells to be able to use them in large scale network simulations. Most often, they aim to achieve this by drastically reducing the number of compartments \citep{Traub2005}. This approach however also changes the spatio-temporal response properties of the nerve cells. With the SGF formalism, inputs that would otherwise be grouped in a small number of compartments may now be grouped at a small number of input locations, while the response properties induced by the neuronal morphology would remain unchanged.  Finally, in dendritic arborizations, a number of resonance phenomena have been observed that can be modeled with linearized (quasi-active) ion channels \citep{Koch1984, Laudanski2014}. Since the transitivity property (see appendix \S\ref{app:VF})) holds for equations of this type as well (of the general form \eqref{eq:PDE}), such systems can be modeled implicitly in the SGF formalism. Nevertheless, the validity of this linearization depends on the size of the fluctuations around the operating point and the ion channel under study and has to be checked.

In the SGF formalism there is a significant initialization phase before a neuron model can be simulated (see Implementation). The computational cost of this phase is higher than the cost of the initialization phase for finite difference approaches, and depends on the complexity of the tree graph and the number of input locations. We thus expect that our SGF formalism will be advantageous in use cases where frequent re-initialization of the model is not required. This is typically the case for network simulations, where a limited number of prototypical nerve cells may be initialized and used throughout the network.

Another important matter, next to computational cost, is the accuracy of the SGF formalism. While the sparsification is exact, the transform back to the time-domain along with the specific integration algorithm might introduce errors. The use of the approximate VF algorithm however impedes a systematic analysis of these errors. Nevertheless, we found that typically this error was very low (Fig~\ref{fig:VF}B,C). Furthermore, after application of the integration paradigm described in this work, where the convolutions with exponentials are integrated analytically assuming that the voltage varies linearly in between grid points, we found that all our numerical experiments agreed very well with equivalent \textsc{neuron} simulations.

\section*{Acknowledgments}
We thank Luc Guyot for many insightful comments on the manuscript. The research leading to these results has received funding from the European Union - Seventh Framework Programme  [FP7/2007-2013] under grant agreements no. 269921 (BrainScaleS) and no. 604102  (The HUMAN BRAIN PROJECT) as well as funding from the EPFL Blue Brain Project Fund and the ETH Board to the Blue Brain Project.

\appendix
\section{Proof of Lemma \ref{thm:transitivity}} \label{app:trans}
\begin{proof}
\textbf{The PDE defined on a line}\\
Consider a line of length $L$ ($0 \leq x \leq L$), which trivially has two leafs ($\Lambda = \{\lambda_1, \lambda_2 \}$). Fourier transforming a PDE such as \eqref{eq:PDE} leads to a boundary value problem of the form:
\begin{equation}\label{eq:boundval}
\fL_0(x,\omega)\frac{\partial^2}{\partial x^2}\fV(x,\omega) + \fL_1(x,\omega)\frac{\partial}{\partial x}\fV(x,\omega) + \fL_2(x,\omega)\fV(x,\omega) = \fI(x,\omega), \hspace{4mm} 0 < x < L \\
\end{equation}
\begin{equation}\label{eq:boundcond}
\begin{aligned}
\hat{B}^{\lambda_1}\fV(0,\omega) \defeq \fL_1^{\lambda_1}(\omega) \frac{\partial}{\partial x} \fV(L,\omega) + \fL_2^{\lambda_1}(\omega) \fV(0,\omega) = \fI^{\lambda_1}(\omega) \\
\hat{B}^{\lambda_2}\fV(0,\omega) \defeq \fL_1^{\lambda_2}(\omega) \frac{\partial}{\partial x} \fV(L,\omega) + \fL_2^{\lambda_2}(\omega) \fV(L,\omega) = \fI^{\lambda_2}(\omega),
\end{aligned}
\end{equation}
The Green's functions $g(x,x_0,\omega)$ is obtained from solving this problem for $\fI(x,\omega)=\delta(x-x_0)$ and $\fI^{\lambda_1}(\omega)=\fI^{\lambda_2}(\omega)=0$, and is given in \citep[page 66]{Stakgold1967} :
\begin{equation}\label{eq:gf}
g(x,x_0) = 
\begin{cases}
\frac{u_1(x)u_2(x_0)}{a_0(x_0)W(u_1,u_2;x_0)}, \hspace{4mm} 0 \leq x \leq x_0 \\
\frac{u_1(x_0)u_2(x)}{a_0(x_0)W(u_1,u_2;x_0)}, \hspace{4mm} x_0 \leq x \leq L,
\end{cases}
\end{equation}
where $u_1(x)$ is a non-trivial solution of the homogeneous equation satisfying $\hat{B}^{\lambda_1} u_1=0$ and $u_2(x)$ a non-trivial solution satisfying $\hat{B}^{\lambda_2} u_2=0$, where $W(u_1,u_2;x)$ denotes the Wronskian of both solutions evaluated at $x$ and where we have omitted the dependence on $\omega$ for clarity. From this equation, it can be checked that the \textbf{transitivity property} holds:
\begin{equation}\label{eq:trans2}
g(x_1,x_3) = \frac{g(x_1,x_2)g(x_2,x_3)}{g(x_2,x_2)}
\end{equation}
when $x_1 \leq x_2 \leq x_3$ or $x_3 \leq x_2 \leq x_1$. \\

\textbf{The generalization to a tree graph} \\
Consider now the generalization of problem \eqref{eq:boundval} to a tree graph. On each edge $\epsilon \in E$, an operator of the form:
\begin{equation}
\hat{L}^{\epsilon}(x) \defeq \fL_0^{\epsilon}(x)\frac{\partial^2}{\partial x^2} + \fL_1^{\epsilon}(x)\frac{\partial}{\partial x}+ \fL_2^{\epsilon}(x)
\end{equation}
constrains the field $V^{\epsilon}$:
\begin{equation}\label{eq:branch}
\hat{L}^{\epsilon}(x)\fV^{\epsilon}(x) = \fI^{\epsilon}(x),
\end{equation}
where it is understood that $x$ signifies the space coordinate on the edge under consideration. On each leaf $\lambda \in \Lambda$ a boundary condition of the form \eqref{eq:boundcond} holds:
\begin{equation}\label{eq:leaf}
\hat{B}^{\lambda} \fV^{\lambda} = \fI^{\lambda}
\end{equation}
and at each node that is not a leaf $\phi \in \Phi$:
\begin{eqnarray} \label{eq:cont} 
& \fV^{\epsilon} = \fV^{\epsilon'}, \hspace{4mm} \forall \epsilon, \epsilon' \in E(\phi)\\ \label{eq:flow}
& \sum_{\epsilon \in E(\phi) }\fL^{\epsilon} \frac{\partial}{\partial x} \fV^{\epsilon} = \fI^{\phi},
\end{eqnarray}
Equation \eqref{eq:cont} assures continuity of the field and a condition of the form \eqref{eq:flow} is imposed in many physical systems to assure the conservation of flow. 

We wish to determine Green's function $g(x,x_0)$ of this problem. When $x_0$ is located on edge $\epsilon_0$, we need to solve this problem for $I^{\epsilon}(x)=\delta_{\epsilon \epsilon_0}\delta(x-x_0), I^{\lambda}=I^{\phi}=0$. We will first show that at each end of edge $\epsilon_0$ boundary conditions of the form \eqref{eq:boundcond} hold, by using \eqref{eq:leaf}, \eqref{eq:cont} and \eqref{eq:flow}.

To do so, we only need the following recursion rule: Consider a node $\phi$, and suppose that all but one of the edges in $E(\phi)$ satisfy a boundary condition $\hat{B}^{\epsilon} \fV^{\epsilon} = 0$ of the type \eqref{eq:boundcond} at the opposite end of node $\phi$. Within each edge, we chose the spatial coordinate $x^{\epsilon}$ to be $L^{\epsilon}$ (i.e. the edge's length) at that far end and 0 at the node. Let us call the edge that does not satisfy the boundary condition $\epsilon'$. Thus we have:
\begin{equation} \label{eq:unmixed}
\hat{B}^{\epsilon} \fV^{\epsilon} = \fL_1^{\epsilon} \frac{\partial}{\partial x} \fV^{\epsilon}(L^{\epsilon}) + \fL_2^{\epsilon} \fV^{\epsilon}(L^{\epsilon}) = 0, \hspace{4mm} \forall \epsilon \in E(\phi) \setminus \epsilon',
\end{equation}
and from \eqref{eq:cont} and \eqref{eq:flow} if follows that:
\begin{eqnarray} \label{eq:connect1}
& \fV^{\epsilon'}(0) = \fV^{\epsilon}(0), \hspace{4mm} \forall \epsilon \in E(\phi)\\ \label{eq:connect2}
& \sum_{\epsilon \in E(\phi)\setminus \epsilon' }\fL^{\epsilon} \frac{\partial}{\partial x} \fV^{\epsilon}(0) + \fL^{\epsilon'} \frac{\partial}{\partial x} \fV^{\epsilon'}(0) = 0,
\end{eqnarray}
We will show that from conditions \eqref{eq:connect1} and \eqref{eq:connect2} a boundary condition of the form \eqref{eq:boundcond} can be derived for $\fV^{\epsilon'}(0)$ (i.e. the field on edge $\epsilon'$ at the location of node $\phi$), as long as the differential equations on edges $\epsilon \in E(\phi)\setminus \epsilon'$ are homogeneous. Let $u^{\epsilon}(x)$ be a non-trivial solution of the homogeneous problem \eqref{eq:branch} on edge $\epsilon$ that satisfies condition \eqref{eq:unmixed}. Every solution on that edge is necessarily of the form $\fV^{\epsilon}(x) = A^{\epsilon}u^{\epsilon}(x)$. As a consequence of \eqref{eq:connect1}, $A^{\epsilon} = \frac{\fV^{\epsilon'}(0)}{u^{\epsilon}(0)}$, which leads to a constraint on the derivative $\frac{\partial}{\partial x}\fV^{\epsilon}(0)=\frac{\fV^{\epsilon'}(0)}{u^{\epsilon}(0)}\frac{\partial}{\partial x}u^{\epsilon}(0)$. Thus, equation \eqref{eq:connect2} becomes:
\begin{equation}
 \left( \sum_{\epsilon \in E(\phi)\setminus \epsilon' } \fL^{\epsilon} \frac{\frac{\partial}{\partial x} u^{\epsilon}(0)}{u^{\epsilon}(0)} \right) \fV^{\epsilon'}(0) + \fL^{\epsilon'} \frac{\partial}{\partial x} \fV^{\epsilon'}(0) = 0,
\end{equation}
precisely the boundary condition for $\fV^{\epsilon'}$ we were after. 

Applying this operation recursively throughout the tree, starting from the leafs until edge $\epsilon_0$, assures that this edge has a boundary condition of the form \eqref{eq:boundcond} at both ends. Thus, on this edge, $g(x,x_0)$ is of the form \eqref{eq:gf}.

To prove the transitivity property \eqref{eq:trans2} for two arbitrary points $x_1$ and $x_3$ on the tree graph, and for a point $x_2$ that is on the shortest path between $x_1$ and $x_3$, we distinguish four cases. 
\begin{itemize}
\item[1.] {\bf$x_1,x_2, x_3$ are on the same edge.} Since the segment has boundary conditions of the type \eqref{eq:boundcond}, the Green's function may be constructed as in \eqref{eq:gf}, and thus \eqref{eq:trans2} holds.

\item[2.] {\bf$x_1$ and $x_2$ on the same edge, $x_3$ is on a different edge.} Let $\phi$ be the node adjacent to the edge $\epsilon_0$ on which $x_1$ and $x_2$ are located and on the shortest path between $x_2$ and $x_3$. Necessarily, the Green's function at that point satisfies
\begin{equation}
g(\phi,x_1) = \frac{g(\phi,x_2)g(x_2,x_1)}{g(x_2,x_2)},
\end{equation}
which then determines the solution on the adjacent edges $\epsilon \in E(\phi) \setminus \epsilon_0$ (where we choose the spatial coordinate $x^{\epsilon}$ to be $0$ at $\phi$ and $L^{\epsilon}$ at the opposite end). On either of these edges, the solution is of the form $A^{\epsilon}u^{\epsilon}(x^{\epsilon})$, with $u^{\epsilon}(x^{\epsilon})$ a solution satisfying the derived condition of type \eqref{eq:boundcond} at the opposite end. Condition \eqref{eq:cont} then imposes $A^{\epsilon} = \frac{g(\phi,x_1)}{u^{\epsilon}(0)} = \frac{g(\phi,x_2)g(x_2,x_1)}{g(x_2,x_2)u^{\epsilon}(0)}$, and thus:
\begin{equation}
\begin{aligned}
 g(x^{\epsilon},x_1) & = \frac{g(\phi,x_1)}{u^{\epsilon}(0)} u^{\epsilon}(x^{\epsilon}) = \frac{g(\phi,x_2)g(x_2,x_1)}{g(x_2,x_2)u^{\epsilon}(0)}u^{\epsilon}(x^{\epsilon}) \\ 
 & = \frac{\frac{g(\phi,x_2)u^{\epsilon}(x^{\epsilon})}{u^{\epsilon}(0)}g(x_2,x_1)}{g(x_2,x_2)}=\frac{g(x^{\epsilon},x_2)g(x_2,x_1)}{g(x_2,x_2)}.
\end{aligned}
\end{equation}
We may apply this consideration recursively through the tree graph, until we arrive at the point $x_3$, which proves relation \eqref{eq:trans2} in this case.

\item[3.] {\bf$x_2$ and $x_3$ on the same edge, $x_1$ is on a different edge.}
Let $\epsilon$ denote the edge on which $x_2$ and $x_3$ are located, and let us denote the node adjacent to that edge at the side of $x_1$ by $\phi$, and take the $x$-coordinate describing the position in that edge to be zero there ($x^{\epsilon}=0$). Then $\fV^{\epsilon}(0)=g(\phi,x_1)$. On the other side of the edge, at $x^{\epsilon}=L^{\epsilon}$, the derived condition of the form \eqref{eq:boundcond} holds, and thus the field in that edge satisfies $\fV^{\epsilon}(x^{\epsilon}) = \frac{g(\phi,x_1)}{u^{\epsilon}(0)}u^{\epsilon}(x^{\epsilon})(=g(x^{\epsilon},x_1))$, where $u^{\epsilon}(x^{\epsilon})$ is a solution satisfying this boundary condition. The Green's functions $g(x_2,x_2)$ and $g(x_3,x_2)$ are still of the form \eqref{eq:gf}, as in this case the derived boundary conditions are valid on both ends of the edge. Let $v^{\epsilon}(x^{\epsilon})$ be a solution that satisfies the homogeneous boundary condition at $x^{\epsilon}=0$. Then, using \eqref{eq:gf}, the following holds:
\begin{equation}
\begin{aligned}
g(x_3,x_2)\frac{1}{g(x_2,x_2)}&g(x_2,x_1) \\
& = \frac{v^{\epsilon}(x_2)u^{\epsilon}(x_3)}{a_0(x_2)W(u^{\epsilon},v^{\epsilon},x_2)}\frac{a_0(x_2)W(u^{\epsilon},v^{\epsilon},x_2)}{v^{\epsilon}(x_2)u^{\epsilon}(x_2)}\frac{g(\phi,x_1)u^{\epsilon}(x_2)}{u^{\epsilon}(0)}\\
& = \frac{g(\phi,x_1)u^{\epsilon}(x_3)}{u^{\epsilon}(0)} \\
& = g(x_3,x_1)
\end{aligned}
\end{equation}
\item[4.] {\bf$x_1,x_2, x_3$ are on different edges.} This case is proved by combining the two previous cases.
\end{itemize}
\end{proof}

\section{Vector fitting}\label{app:VF}
Here we briefly explain the VF algorithm \citep{Gustavsen1999} as we implemented it. The version of this algorithm we needed approximates a complex function $f(s)$, for which $\abs{f(s)} \rightarrow 0$ when $\abs{s} \rightarrow \infty$ and $\abs{f(s)} > 0$, as follows:
\begin{equation}
f(s) \approx \sum_{l=1}^L \frac{\gamma_l}{\alpha_l + s},
\end{equation}
where the parameter $L$ is chosen. It does so in two steps: First the poles $\alpha_l$ are identified and then the residues $\gamma_l$ are fitted. \\

\textbf{Pole identification} \\
First, a set of chosen starting poles $\bar{\alpha}_l$ is specified, and an unknown auxiliary function $\sigma(s)$ with these poles is proposed, so that:
\begin{align} \label{eq:sigma}
\sigma(s) = \sum_{l=1}^L \frac{\bar{\gamma}_l}{\bar{\alpha}_l + s} + 1 \\ \label{eq:sigmaf}
\left[ \sigma(s)f(s) \right]_{\text{fit}} = \sum_{l=1}^L \frac{\bar{\gamma}_l'}{\bar{\alpha}_l + s}.
\end{align}
Multiplying equation \eqref{eq:sigma} with $f(s)$, and equating this with equation \eqref{eq:sigmaf} gives:
\begin{equation}\label{eq:overdet}
\sum_{l=1}^L \frac{1}{\bar{\alpha}_l + s}\bar{\gamma}_l' - \sum_{l=1}^L \frac{f(s)}{\bar{\alpha}_l + s}\bar{\gamma}_l = f(s).
\end{equation}
When $f$ is evaluated at enough frequency points $s$ (we found it sufficient to sample $f(s)$ on the imaginary axis $s \equiv i\omega_j, \omega_j \in \mathbb{R}, j=1,\hdots,N$ on an equidistant scale for small $\omega$ and on a logarithmic scale for large $\omega$), this gives an over-determined system that can be solved for $\bar{\gamma}_l$ and $\bar{\gamma}_l'$ by the least squares method. The poles of $f(s)$ are then given by the zeros of $\sigma(s)$ since, as the parametrization for $\sigma(s)$ satisfies equation \eqref{eq:overdet}, it holds that (see \citep{Gustavsen1999, Hendrickx2006} for further details):
\begin{equation}
f(s) = \frac{\left[ \sigma(s)f(s) \right]_{\text{fit}}}{\sigma(s)}.
\end{equation}
These zeros can be found as the eigenvalues of the matrix:
\begin{equation}
H = \left(
\begin{array}{cccc}
\bar{\alpha}_1 - \bar{\gamma}_1 &  - \bar{\gamma}_2 & \hdots & - \bar{\gamma}_L \\
- \bar{\gamma}_1 & \bar{\alpha}_2 - \bar{\gamma}_2 & \hdots & - \bar{\gamma}_L \\
\vdots & \vdots & \ddots & \vdots \\
- \bar{\gamma}_1 & - \bar{\gamma}_2 & \hdots & \bar{\alpha}_L - \bar{\gamma}_L
\end{array} \right).
\end{equation}
Note that this procedure can be part of an iterative optimization, where the newly found poles can be used as the starting poles for the next iteration. \\

\textbf{Residue fitting} \\
Once the poles $\alpha_l$ are known, the residues can be determined by solving the over-determined system:
\begin{equation}
\sum_{l=1}^L \frac{1}{\alpha_l + s}\gamma_l = f(s)
\end{equation}
for $\gamma_l$ by the least squares method. \\

\begin{figure}
   \centering
   \includegraphics[width=\textwidth]{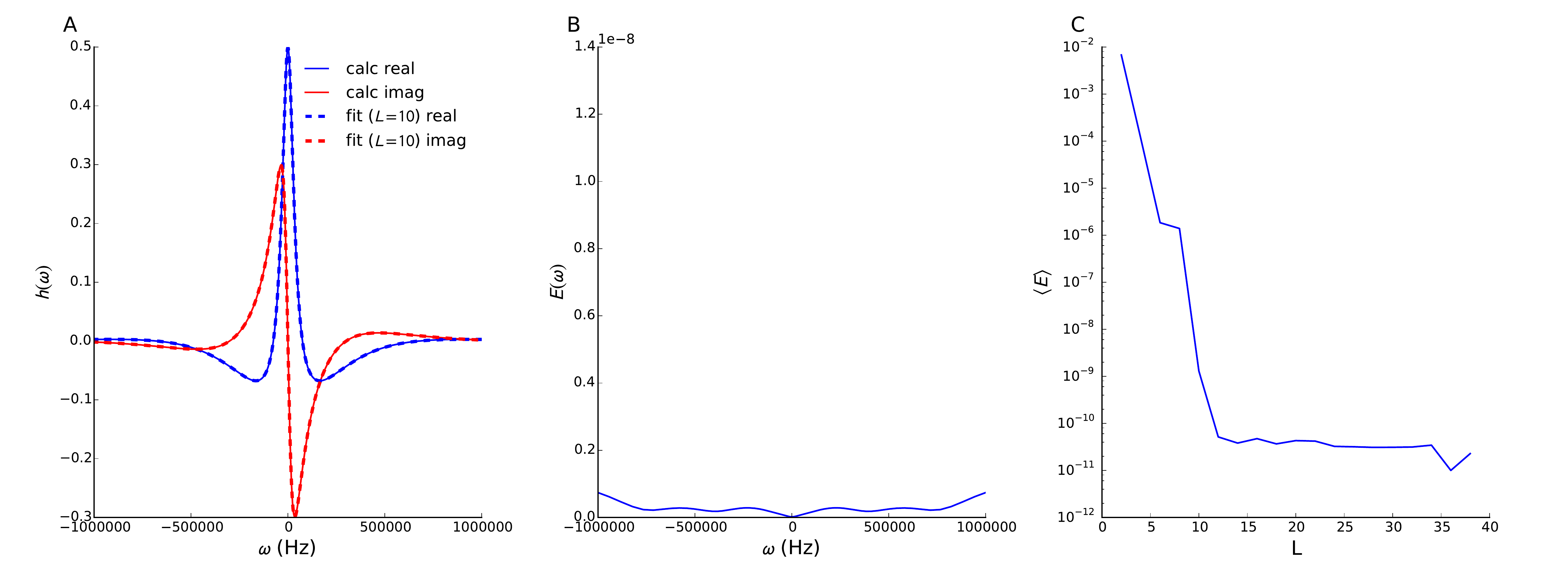}
   \caption{\textbf{Illustration of the VF algorithm.} A: A typical kernel is accurately approximated by the VF algorithm (here $L=10$). B: The error $E(\omega)$ of this fit. C: The average error of the fit as a function of $L$.}
   \label{fig:VF}
\end{figure}

\textbf{Accuracy} \\
The VF algorithm does not provide the number of poles for the fit, nor does it provide an estimate for the accuracy of the fit with a given number of poles. We chose the number of poles as the smallest number for which the approximation gave a sufficient accuracy, defined as:
\begin{equation}
E(\omega_j) \equiv \frac{\abs{f(i\omega_j) -  \sum_{l=1}^L \frac{\gamma_l}{\alpha_l + i\omega_j}}}{\max_j\abs{f(i\omega_j)}} < \epsilon, \hspace{4mm} j=1,\hdots,N.
\end{equation}
For our nerve models, we found that $\epsilon = 10^{-8}$ was sufficient. Usually, this accuracy was reached with $10 \lesssim L \lesssim 20$. In Fig~\ref{fig:VF}A, we show a typical example of a kernel from the SGF formalism, together with its approximation with $L=10$. In Fig~\ref{fig:VF}B we show the error of this approximation as a function of the frequency, whereas Fig~\ref{fig:VF}C we show how the average error, defined as:
\begin{equation}
\langle E \rangle = \frac{1}{N} \sum_{j=1}^{N} E(\omega_j),
\end{equation}
changes with $L$.


\end{document}